\documentclass{article}

\usepackage{microtype}
\usepackage{graphicx}
\usepackage{subfigure}
\usepackage{booktabs} 

\usepackage{hyperref}

\usepackage[accepted]{icml2024}

\usepackage{amsmath}
\usepackage{amssymb}
\usepackage{mathtools}
\usepackage{amsthm}
\usepackage{bm}
\usepackage{subcaption}
\usepackage{multicol}
\usepackage{comment}

\usepackage[capitalize,noabbrev]{cleveref}

\theoremstyle{plain}
\newtheorem{theorem}{Theorem}[section]

\newtheorem{lemma}[theorem]{Lemma}

\theoremstyle{definition}
\newtheorem{definition}[theorem]{Definition}

\theoremstyle{remark}

\DeclareMathOperator*{\argmin}{argmin}

\usepackage[textsize=tiny]{todonotes}

\newcommand{\subf}[2]{%
  {\small\begin{tabular}[tb]{@{}c@{}}
  #1\\#2
  \end{tabular}}%
}

\newcommand{\defeq}{\vcentcolon=}

\icmltitlerunning{Theoretical and Empirical Analysis of Adaptive Entry Point Selection for Graph-based ANNS}

\begin{document}

\twocolumn[
\icmltitle{Theoretical and Empirical Analysis of Adaptive Entry Point Selection for Graph-based Approximate Nearest Neighbor Search}



\icmlsetsymbol{equal}{*}

\begin{icmlauthorlist}
\icmlauthor{Yutaro Oguri}{yyy}
\icmlauthor{Yusuke Matsui}{yyy}
\end{icmlauthorlist}

\icmlaffiliation{yyy}{The University of Tokyo, Tokyo, Japan}

\icmlcorrespondingauthor{Yutaro Oguri}{oguri@hal.t.u-tokyo.ac.jp}

\icmlkeywords{Approximate Nearest Neighbor Search, Graph-based Index}

\vskip 0.3in
]



\printAffiliationsAndNotice{}  

\begin{abstract}

We present a theoretical and empirical analysis of the adaptive entry point selection for graph-based approximate nearest neighbor search (ANNS). We introduce novel concepts: $b$\textit{-monotonic path} and $B$\textit{-MSNET}, which better capture an actual graph in practical algorithms than existing concepts like MSNET. We prove that adaptive entry point selection offers better performance upper bound than the fixed central entry point under more general conditions than previous work. Empirically, we validate the method's effectiveness in accuracy, speed, and memory usage across various datasets, especially in challenging scenarios with out-of-distribution data and hard instances. Our comprehensive study provides deeper insights into optimizing entry points for graph-based ANNS for real-world high-dimensional data applications.

\end{abstract}

\section{Introduction}
\label{sec:introduction}
Nearest Neighbor Search (NNS) is an algorithm that searches for the vector closest to a query vector within a set of vectors. Approximate Nearest Neighbor Search (ANNS)~\cite{arya1993approximate_ANNS} approximates NNS, achieving faster search speeds while sacrificing a small quantity of accuracy. ANNS is vital in various domains like image retrieval~\cite{lowe2004distinctive_image_features_SIFT}. ANNS becomes essential for high-dimensional vectors and large-scale data~\cite{aumuller2020ann-benchmarks}. We evaluate the performance of ANNS algorithms based on accuracy, speed, and memory usage.


Among ANNS methods, the graph-based approach offers the best trade-off between accuracy and performance, as long as the data fits in memory~\cite{cvpr23_tutorial_neural_search,Fu2017FastAN_NSG,malkov2018efficient_HNSW}. The graph-based index constructs a graph with the vectors in the database as nodes. One can start the search from an entry point and traverse the graph towards the query to find the nearest vector.

Recent research demonstrated that the entry points are significant regarding the performance~\cite{arai2021lgtm,iwasaki2018optimization_NGT,oguri2023generalSisap,ni2023diskann++}. When a query and the entry point are distant, the number of hops on the path increases, and the search speed decreases. Additionally, we may not reach the nearest neighbors of distant queries, causing a decrease in accuracy. Oguri and Matsui proposed adaptively selecting the entry point from candidates generated by k-means clustering~\cite{oguri2023generalSisap}. Such adaptive selection is empirically known to enable efficient search. In a subsequent work~\cite{ni2023diskann++} about ANNS on a memory-disk hybrid setting, Ni et al. proposed the same method and demonstrated its benefit.

However, it remains unclear when and why the adaptive entry point selection is effective. This research provides theoretical and empirical analysis to address this issue. Our contributions are as follows:

\textbf{Theoretical Aspect}
\begin{itemize}
    \item We introduced a novel concept of $b$\textit{-monotonic path} and $B$\textit{-MSNET} that generalize the existing concept of MSNET~\cite{dearholt1988monotonic_MSNET}. It better captures actual graphs in practical algorithms.
    \item We proved that the adaptively selected entry point works better than the fixed central point. Compared to proof in previous research, it is simpler and is valid under more general conditions on data and graphs.
\end{itemize}
\textbf{Empirical Aspect}
\begin{itemize}
    \item We extensively evaluated the adaptive entry point selection on datasets with various characteristics, such as out-of-distribution (OOD) settings. We showed that the method improved the speed of NSG~\cite{Fu2017FastAN_NSG} about 1.2 - 2.3 times. We also demonstrated its effectiveness in terms of accuracy and memory usage.
    \item We demonstrated that the adaptive entry point selection is effective against hard instances proposed in the recent study~\cite{indyk2023worstcase} that significantly reduce the performance of the graph-based index. We empirically examined the reasons for this effectiveness.
\end{itemize}

\section{Related Work}
\subsection{Approximate Nearest Neighbor Search (ANNS)}
Given a database consisting of $N$ $d$-dimensional vectors $\mathcal{X} = \{\bm{x}_1, \dots, \bm{x}_N\} \subset \mathbb{R}^d$ and a query vector $\bm{q}\in\mathbb{R}^d$, Nearest Neighbor Search (NNS) is an algorithm to answer the closest vector $\bm{x}^* = \argmin_{\bm{x}\in\mathcal{X}}\Vert\bm{q} - \bm{x}\Vert_2$. An extension of NNS that returns the $\text{top-}K$ closest vectors is $K-$Nearest Neighbor Search ($K-$NNS). Since it requires calculating the distance between every vector in the database and the query, it is impractical for large-scale, high-dimensional datasets as typically seen in real applications~\cite{li2019anns_highdim_dataexp_analysis,aumuller2020ann-benchmarks}. Therefore, Approximate Nearest Neighbor Search (ANNS) uses an index to speed up the search significantly. We can divide algorithms for into four types: tree-based~\cite{muja2014scalableFLANN,silpa2008optimised_KDTree}, quantization-based~\cite{ge2013optimized,Jegou2011ProductQuantization,huijben2024residual_neural_quantization}, hash-based~\cite{gionis1999similarityLSH,andoni2015practical_and_optimal_LSH_angular}, and graph-based~\cite{wang2021comprehensive_survey_graph_2021,malkov2018efficient_HNSW,Fu2017FastAN_NSG,jayaram2019diskann}. Each has a performance tradeoff.


Applications of ANNS include image retrieval~\cite{lowe2004distinctive_image_features_SIFT}, recommendation systems~\cite{suchal2010full_recommendation,chen2022approximate_large_recommendation}, and knowledge augmentation in Large Language Models~\cite{asai2023retrieval_based_LLM}.

\subsection{Graph-based Index}

\begin{algorithm}[tb]
   \caption{Search on Graph-based Index}
   \label{alg:graph_search}
\begin{algorithmic}
    \STATE {\bfseries Input:} Index $G(\mathcal{V}, \mathcal{E})$, Entry Point $v_s\in\mathcal{V}$, Query $\bm{q}\in\mathbb{R}^d$, Length of Search Queue $L\in\mathbb{Z}$
    \STATE {\bfseries Output:} Nearest Node to $\bm{q}$ in $C$
    \STATE Initialize Candidate Queue $C = \{v_s\}$
    \STATE Initialize Visited Node Set $T = \{v_s\}$
    \REPEAT
        \STATE $u = C$.popNearestNode()
        \FOR{$v \in G$.getNeighbors(u)}
            \IF{$v \notin T$}
            \STATE $C$.enqueue($v$)
            \STATE $T = T \cup \{v\}$
            \IF{$\vert C \vert > L$}
                \STATE $C$.popFarthestNode()
            \ENDIF
            \ENDIF
        \ENDFOR
    \UNTIL{$C$ not updated}
\end{algorithmic}
\end{algorithm}

We empirically know that the graph-based index offers the best trade-off between accuracy and speed for million-scale problems where the entire database fits in memory~\cite{cvpr23_tutorial_neural_search,Fu2017FastAN_NSG,malkov2018efficient_HNSW}. It constructs a graph $G(\mathcal{V}, \mathcal{E})$ by corresponding nodes $v_i\in\mathcal{V}$ to vectors in the database $\bm{x}_i\in\mathcal{X}$, and search for the nearest neighbors on the graph from an entry point (\cref{alg:graph_search}).

The graph index approximates a base graph that has specific characteristics such as Delaunay Graph, Relative Neighborhood Graph~\cite{jaromczyk1992relative_RNG}, and Monotonic Search Network (MSNET)~\cite{dearholt1988monotonic_MSNET}. NSG~\cite{Fu2017FastAN_NSG} inherits the characteristics of MSNET approximately, but note that it requires a strong assumption that a query is included in a database~\cite{Fu2017FastAN_NSG,prokhorenkova2020graph_practice_theory}. It refines a constructed KNN graph~\cite{paredes2005using_knn_graph} by Efanna~\cite{fu2016efanna} or NNDescent~\cite{dong2011efficient_Kgraph_NNDescent}. DiskANN~\cite{jayaram2019diskann} targets the memory-disk hybrid settings. HNSW~\cite{malkov2018efficient_HNSW} constructs hierarchical graphs and is one of the SoTA implementations. Note that the adaptive entry point selection does not target such hierarchical indexes. In practical application, HNSW, NSG, and DiskANN are the most widely used methods~\cite{Fu2017FastAN_NSG,zhang2022uni_bing_search}. 


\subsection{Characteristics of dataset in ANNS}
We use a variety of datasets for evaluating ANNS algorithms. Traditionally, datasets composed of feature descriptors like SIFT~\cite{lowe2004distinctive_image_features_SIFT} and GIST~\cite{oliva2001modeling_GIST_original} served as standard benchmarks. Recently, with the advancement of deep learning, it is expected to evaluate the algorithm on neural feature vectors~\cite{Simhadri2022NeuripsComp}. It becomes common to conduct cross-modal searches like Text-to-Image retrieval, using image embeddings as the database and text embeddings as the query. These datasets exhibit differences in the statistical properties of queries and databases due to modality. It is a class of Out-Of-Distribution (OOD) query setting. A previous work~\cite{jaiswal2022ood_diskann} demonstrated that their performance significantly decreases compared to ordinary settings.

Recent research has made progress in theoretically analyzing graph-based indexes~\cite{prokhorenkova2020graph_practice_theory}. A previous work~\cite{indyk2023worstcase} analyzes the worst-case performance and provides hard instances where graph-based indexes achieve near-worst performance.


\section{Preliminary}

\begin{table}[tb]
\caption{Mathematical Notations}
\label{table:notations}
\vskip 0.15in
\begin{center}
\begin{small}
\begin{tabular}{@{}ll@{}}
\toprule
Notations & Descriptions \\
\midrule
$d$ & Dimensionality of vectors. \\
$N$ & Size of a database. \\
$\mathcal{X} \subset \mathbb{R}^d$ & A database $\mathcal{X} = \{\bm{x}_1, \dots, \bm{x}_N\} \subset \mathbb{R}^d$. \\
$\Vert\cdot\Vert_2$ & Euclidean norm of a vector. \\
$G(\mathcal{V}, \mathcal{E})$ & A graph $G$ with vertices $\mathcal{V}$ and edges $\mathcal{E}$. \\
$\phi: \mathcal{X} \to \mathcal{V}$ & A one-to-one mapping from $\bm{x}_i\in\mathcal{X}$ to $v_i \in \mathcal{V}$. \\
$\bm{\phi}^{-1}: \mathcal{V}\to\mathcal{X}$ & An inverse mapping of $\phi$. \\
$\mathcal{P}(v_s, v_t) \subseteq \mathcal{V}$ & A path on a graph from $v_s$ to $v_t$. \\
$\textbf{NN}(\bm{x}, \mathcal{S})$ & $\argmin_{\bm{y}\in\mathcal{S}}\Vert\bm{x} - \bm{y}\Vert_2$. \\
$\textbf{GT}(\bm{q})$ & $\textbf{NN}(\bm{q}, \mathcal{X})$. A ground truth vector for $\bm{q}$. \\
\bottomrule
\end{tabular}
\end{small}
\end{center}
\vskip -0.1in
\end{table}

Let us first introduce the preliminary knowledge. \cref{table:notations} lists mathematical notations used in the paper.

\subsection{Voronoi Partition} \label{sec:voronoi_partition}
This section introduces the Voronoi partition used in our theoretical analysis and empirical explanations. Let $\mathcal{U} \subset \mathbb{R}^d$ be a finite subset of Euclidean space $\mathbb{R}^d$, and $K$ be the number of cells in the Voronoi partition. The Voronoi partition depends only on a set of representative points called `sites' $\mathcal{D} = \{\bm{d}_1, \dots, \bm{d}_K\}\subset\mathbb{R}^d$. For each site $\bm{d}_j \in \mathcal{D}$, we define a Voronoi cell $\mathcal{U}_j$ as follows:
\begin{align}
    \mathcal{U}_j = \{\bm{y} \in \mathcal{U} \mid \argmin_{\bm{d}\in\mathcal{D}}\Vert\bm{y}-\bm{d}\Vert_2 = \bm{d}_j\}.
\end{align}
Then, we can divide $\mathcal{U}$ into $K$ cells $\{\mathcal{U}_1, \dots, \mathcal{U}_K\}$ and we have $\mathcal{U} = \bigcup_{j=1}^{K}\mathcal{U}_j$. We assume the cells are disjoint each other for simplicity.

Let $\mathcal{U}(\mathcal{X})$ be a finite region that include $\mathcal{X}$ and all queries $\bm{q}\in\mathbb{R}^d$. We define Voronoi partition on $\mathcal{U}(\mathcal{X})$, and denote each Voronoi cells by $\mathcal{U}_j \subset \mathcal{U}(\mathcal{X})$.

\subsection{The fixed central entry point}
NSG~\cite{Fu2017FastAN_NSG} and DiskANN~\cite{jayaram2019diskann} start the search from the fixed central entry point. We denote the central point $\bm{d}_0$ of the database by:
\begin{align}
    \bm{d}_0 &= \textbf{NN}\left(\frac{1}{\vert\mathcal{X}\vert}\sum_{\bm{x}\in\mathcal{X}} \bm{x}, \mathcal{X}\right).
\end{align}

\subsection{Recap of Adaptive Entry Point Selection} \label{sec:recap_ep_selection}
We review the adaptive entry point selection using k-means clustering~\cite{oguri2023generalSisap,ni2023diskann++}.

\textbf{Generating entry point candidates}\hspace{0.1in}We obtain a set of entry point candidates by dividing the entire database $\mathcal{X}$ through clustering and computing the nearest neighbor vector to each cluster center. We perform k-means clustering~\cite{Lloyd1982LeastSQkmeanslloyd} on the database $\mathcal{X} $ to obtain $ K $ clusters and their cluster centers $ \mathcal{C} = \{\bm{c}_1, \dots, \bm{c}_K\} \subset \mathbb{R}^d $. Next, for each cluster center $ \bm{c}_i \in \mathcal{C} $, we compute the nearest neighbor $ \bm{d}_i = \argmin_{\bm{x}\in\mathcal{X}}\Vert\bm{c}_i - \bm{x}\Vert_2 $. Let $ \mathcal{D} = \{\bm{d}_1, \dots, \bm{d}_K\} $ be the set of candidate entry points. The reason to construct $\mathcal{D}$ is that we cannot create a node for $\bm{c}\in\mathcal{C}$ because $\bm{c}\notin\mathcal{X}$.

The time complexity of the phase to generate candidates is $\mathcal{O}(N_{iter} K N d)$ under the fixed number of iteration $N_{iter}$ in k-means. Note that we used highly optimized and significantly faster implementations of k-means like \texttt{Faiss}~\cite{johnson2019billion_faiss,douze2024faiss_library}.

The actual data to be stored is only the candidate set $\mathcal{D}$, and the final space complexity is $ \mathcal{O}(Kd) $. Here, the number of clusters $ K $ is at most about $1000$, which is sufficiently small compared to the memory consumption of the index itself.

\textbf{Select an entry point for a query}\hspace{0.1in}Given a query $\bm{q}\in\mathbb{R}^d$, we first select the entry point $\bm{d}_j = \argmin_{\bm{d}\in\mathcal{D}}\Vert\bm{q} - \bm{d}\Vert_2$. This part involves Brute-Force searching, so the time complexity is $\mathcal{O}(Kd)$, which becomes the search overhead. Subsequently, the search in the graph index continues from the selected entry point $\bm{d}_j$. Increasing the hyperparameter $K$ improves the search speed on the graph index, but the overhead of selecting entry points also increases. That is a trade-off in the whole performance.

\section{Theoretical Analysis} \label{sec:theory}
We introduce two new concepts $b$\textit{-monotonic path} and $B$\textit{-MSNET} to represent a graph-based index in a more general perspective than existing MSNET~\cite{dearholt1988monotonic_MSNET}. In addition, we prove that the adaptively selected entry point offers a better upper bound of performance than the fixed central entry point, assuming that a graph-based index belongs to a class of $B$-MSNET.

\subsection{$b$-monotonic path \& $B$-MSNET}
We introduce two novel concepts, $b$\textit{-monotonic path} and $B$\textit{-MSNET}. They are a generalization of monotonic path monotonic search network (MSNET) that appears in the theoretical background of NSG~\cite{Fu2017FastAN_NSG} and theoretical analysis of graph-based index~\cite{prokhorenkova2020graph_practice_theory}.

Let $\mathcal{P}(v_s, v_t) = \{v_1, \dots, v_{l+1}\}$ $(v_1=v_s~\text{and}~v_{l+1}=v_t)$ be a $l$-hop path from $v_s$ to $v_t$. For each $i \in \{1, \dots, l+1\}$, we denote that the corresponding vector for node $v_i$ as $\bm{x}_i = \bm{\phi}^{-1}(v_i)$. We define
\begin{align}
    r_i = \Vert \bm{x}_{i} - \bm{x}_t \Vert_2 - \Vert \bm{x}_{i+1} - \bm{x}_t \Vert_2.
\end{align}
$r_i$ represents how much the distance to the end $v_t$ changes due to a one-hop from node $v_i$ to $v_{i+1}$. We denote that the collection of $r_i$ defined on the path $\mathcal{P}(v_s, v_t)$ as $ \mathcal{R}(\mathcal{P}) = \{r_1, \dots, r_l\}$. Then, we split $\mathcal{R}(\mathcal{P})$ into two subsets:
\begin{align}
    \mathcal{R}(\mathcal{P}) &= \mathcal{R}(\mathcal{P})_{+} \cup \mathcal{R}(\mathcal{P})_{-} \label{eq:R_path} \\
    \mathcal{R}(\mathcal{P})_{+} &= \{r\in\mathcal{R}(\mathcal{P}) \mid r\geq0\} \label{eq:R_path_plus} \\
    \mathcal{R}(\mathcal{P})_{-} &= \{r\in\mathcal{R}(\mathcal{P}) \mid r < 0\}. \label{eq:R_path_minus}
\end{align}

\begin{figure}[tb]
\vskip 0.2in
    \centering
    \includegraphics[clip, width=0.8\columnwidth]{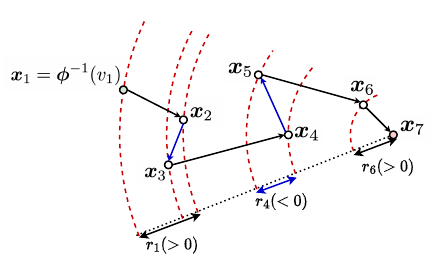}
    \caption{An illustrated example of $b$-monotonic path $(b=2)$ $\mathcal{P}(v_1, v_7) = \{v_1, \dots, v_7\}$ on a graph $G(\mathcal{V}, \mathcal{E})$. Each node $v_i$ corresponds to a vector $\bm{x}_i\in\mathcal{X}$. Note that $r_i = \Vert \bm{x}_{i} - \bm{x}_7 \Vert_2 - \Vert \bm{x}_{i+1} - \bm{x}_7 \Vert_2$ for $i\in\{1, \dots, 6\}$. $r_2, r_4$ are negative and other all $r_i$ are positive. Thus, this path is a $2-$monotonic path. An arrow between two nodes colored with blue represents a backward hop with a negative $r_i$.}
    \label{fig:b-monotonic-example}
\vskip -0.2in
\end{figure}

\begin{figure*}[tb]
\vskip 0.2in
\centering
    \subfigure[$\bm{q}$ and $\textbf{GT}(\bm{q})$ in the same cell]{%
        \includegraphics[clip, width=0.83\columnwidth]{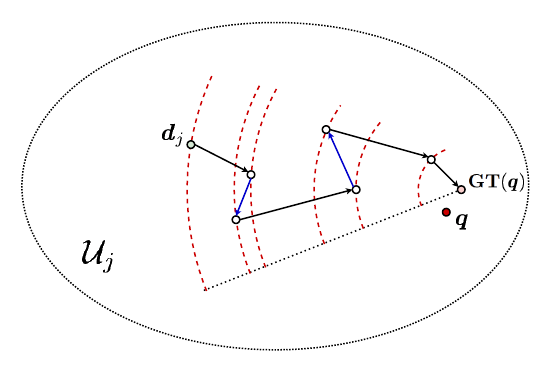}}%
    \hspace{1.4cm}
    \subfigure[$\bm{q}$ and $\textbf{GT}(\bm{q})$ in the different cells]{%
        \includegraphics[clip, width=0.83\columnwidth]{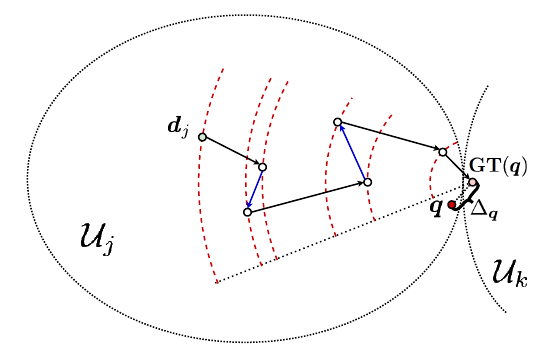}}%
    \caption{A $2-$monotonic path starting from $\bm{d}_j$ to $\textbf{GT}(\bm{q})$. (a) shows the case (i) where $\bm{q}, \textbf{GT}(\bm{q})\in\mathcal{U}_j$, and (b) shows the case (ii) where $\textbf{GT}(\bm{q})\in\mathcal{U}_k$ $(j\neq k)$.}
    \label{fig:b-monotonic-example-cluster}
\vskip -0.2in
\end{figure*}

Then, we define $b$\textit{-monotonic path}.
\begin{definition}[$b$\textit{-monotonic path}]
\label{def:b-monotonic-path}
Let $G(\mathcal{V}, \mathcal{E})$ be a graph-based index constructed on a database $\mathcal{X}$. Let $\mathcal{P}(v_s, v_t) = \{v_1, \dots, v_{l+1}\}$ $(v_1=v_s, v_{l+1}=v_t)$ be a $l$-hop path on $G(\mathcal{V}, \mathcal{E})$ from $v_s$ to $v_t$. The path $\mathcal{P}$ is a $b$-monotonic path, iff $\vert\mathcal{R}(\mathcal{P})_{-}\vert = b$ holds.
\end{definition}
$b$-monotonic path is a generalized concept of monotonic path~\cite{Fu2017FastAN_NSG,ni2023diskann++}. It includes $b$ \textit{backward} steps out of all $l$ steps that go away from the goal, and the rest of $l-b$ steps are \textit{forward} steps that proceed to the goal. \cref{fig:b-monotonic-example} shows an example of $b$-monotonic path $(b=2)$ with $6$-hops.  \cref{fig:b-monotonic-example} shows the radii of concentric spheres centered around $\bm{\phi}^{-1}(v_t)$ colored with red dotted arcs. Considering them, \cref{lemma:b-monotonic-sum-of-r} intuitively follows.
\begin{lemma}
\label{lemma:b-monotonic-sum-of-r}
Let $\mathcal{P}(v_s, v_t)$ be a $b$-monotonic path. Let $\bm{x}_s = \bm{\phi}^{-1}(v_s)$ and $\bm{x}_t = \bm{\phi}^{-1}(v_t)$. The following formula holds:
\begin{equation}
    \Vert \bm{x}_s - \bm{x}_t \Vert_2 = \sum_{r\in\mathcal{R}(\mathcal{P})} r = \sum_{r\in\mathcal{R}_{+}}r + \sum_{r\in\mathcal{R}_{-}}r.
\end{equation}
\end{lemma}

Based on \cref{def:b-monotonic-path}, we introduce a new concept, $B$-Monotonic Search Network ($B$-MSNET).
\begin{definition}[$B$\textit{-MSNET}]
\label{def:B-MSNET}
Given a graph-based index $G(\mathcal{V}, \mathcal{E})$ constructed on a database $\mathcal{X}$, $G(\mathcal{V}, \mathcal{E})$ is a $B$-MSNET iff for any two nodes $v_s, v_t \in \mathcal{V}$, there exists an integer $b\leq B$ and a $b$-monotonic path $\mathcal{P}(v_s, v_t)$.
\end{definition}
$B$-MSNET is a generalized concept of MSNET~\cite{dearholt1988monotonic_MSNET}. When $B=0$, it is equivalent to MSNET.

\subsection{Effectiveness of Adaptive Entry Point Selection}
We introduce a core theorem that supports the effectiveness of entry point selection.
\begin{theorem}
\label{thm:upperbound_proof}
Let $G(\mathcal{V}, \mathcal{E})$ be a $B$-MSNET (a graph-based index on $\mathcal{X}$). Let $\bm{q}\in\mathcal{U}(\mathcal{X})$ be a query. Let $\mathcal{P}(v_s, v_t)$ be a $b$-monotonic path $(b \leq B)$ from $v_s$ to $v_t$, where $v_s$ is the selected entry point for query $\bm{q}$ and $v_t$ is a corresponding node of ground truth $\phi(\textbf{GT}({\bm{q}}))$. Let $\Bar{l}$ be the upper bound of the number of hops of $\mathcal{P}$.

On the other hand, we consider a path $\mathcal{P}_0$ starting from the fixed central point $\phi(\bm{d}_0)$ to $\phi(\textbf{GT}({\bm{q}}))$. Let $\Bar{l}_0$ be the upper bound of the number of hops of $\mathcal{P}_0$.

When at least one of the following two conditions is met, $\Bar{l} \leq \Bar{l}_0$ holds.
\begin{enumerate}
    \item[(i)] $\bm{q}\in\mathcal{U}_j$ $\land$ $\textbf{GT}(\bm{q})\in\mathcal{U}_j$
    \item[(ii)] $\bm{q}\in\mathcal{U}_j$ $\land$ $\textbf{GT}(\bm{q})\notin\mathcal{U}_j$ $\land$ $\Vert\bm{q} - \textbf{GT}(\bm{q})\Vert_2 \leq \Bar{R} - \Bar{R}_j$
\end{enumerate}
where
\begin{align}
    \Bar{R}_j = \max_{\bm{x},\bm{y}\in\mathcal{U}_j}\Vert\bm{x}-\bm{y}\Vert_2 \\
    \Bar{R} = \max_{\bm{x},\bm{y}\in\mathcal{U}(\mathcal{X})}\Vert\bm{x}-\bm{y}\Vert_2.
\end{align}
\end{theorem}

\begin{proof}[Proof of \cref{thm:upperbound_proof}]
We prove the theorem by dividing into (i) and (ii). The visualized example of $b$-monotonic path of these two cases are \cref{fig:b-monotonic-example-cluster}. 

\textbf{(i)} $\bm{q}\in\mathcal{U}_j$ $\land$ $\textbf{GT}(\bm{q})\in\mathcal{U}_j$

We assume the condition (i) in \cref{thm:upperbound_proof} is met. We denote the selected entry point $\bm{d}_j\in\mathcal{U}_j$, which is a site of Voronoi cell $\mathcal{U}_j$. Note that $v_s = \phi(\bm{d}_j)$ and $v_t=\phi(\textbf{GT}(\bm{q}))$.

We define the set of all $b $-monotonic paths ($b \leq B$) where both the starting and ending points are contained in the cell $\mathcal{U}_j$ as $\mathcal{T}(\mathcal{U}_j, B) $:
\begin{equation}
\mathcal{T}(\mathcal{U}_j, B) = \{
\mathcal{P}(\phi(\bm{x}_s), \phi(\bm{x}_t)) \mid
\bm{x}_s, \bm{x}_t \in \mathcal{U}_j
\}
\end{equation}

Moreover, we define the following metrics for a cell $\mathcal{U}_j$:
\begin{align}
    \Bar{r}_{+, j} &= \min_{
    \substack{
        \mathcal{P}\in\mathcal{T}(\mathcal{U}_j, B), \\
        r\in\mathcal{R}(\mathcal{P})_{+}
        }
    } r \\
    \Bar{r}_{-, j} &= \max_{
    \substack{
        \mathcal{P}\in\mathcal{T}(\mathcal{U}_j, B), \\
        r\in\mathcal{R}(\mathcal{P})_{-}
        }
    } |r|.
\end{align}

Considering $\bm{d}_j, \textbf{GT}(\bm{q})\in\mathcal{U}_j$ and \cref{lemma:b-monotonic-sum-of-r}, the following inequality holds:
\begin{align}
    \Bar{R}_j &\geq \Vert\bm{d}_j - \mathbf{GT}(\bm{q})\Vert_2 \\
    &= \sum_{r\in\mathcal{R}(\mathcal{P})_{+}} r + \sum_{r\in\mathcal{R}(\mathcal{P})_{-}} r \\
    &\geq (l - b)\Bar{r}_{+, j} - b \Bar{r}_{-, j}. \label{eq:local_cell_eval}
\end{align}
We evaluate $l$ by an upper bound using $b \leq B$ as follows:
\begin{align}
    l &\leq \frac{\Bar{R}_j}{\Bar{r_{+, j}}} + b\left(1 + \frac{\Bar{ r}_{-, j}}{\Bar{ r}_{+, j}}\right) \\
    &\leq \frac{\Bar{R}_j}{\Bar{r_{+, j}}} + B\left(1 + \frac{\Bar{ r}_{-, j}}{\Bar{ r}_{+, j}}\right) \defeq  \Bar{l} \label{eq:lbar}
\end{align}
Note that $B$ is a constant, and $\Bar{l}$ only depends on the cell $\mathcal{U}_j$, not on each path.

We define the same metrics for $\mathcal{U}(\mathcal{X})$:
\begin{align}
    \Bar{r}_{+} &= \min_{
    \substack{
        \mathcal{P}\in\mathcal{T}(\mathcal{U}(\mathcal{X}), B), \\
        r\in\mathcal{R}(\mathcal{P})_{+}
        }
    } r \label{eq:global_r_plus} \\
    \Bar{r}_{-} &= \max_{
    \substack{
        \mathcal{P}\in\mathcal{T}(\mathcal{U}(\mathcal{X}), B), \\
        r\in\mathcal{R}(\mathcal{P})_{-}
        }
    } |r|. \label{eq:global_r_minus}
\end{align}
Using \cref{eq:global_r_plus,eq:global_r_minus}, in the same way as \cref{eq:local_cell_eval}, we evaluate the number of hops $l_0$ in a path starting from the fixed central point $\phi(\bm{d}_0)$ as follows:
\begin{align}
    l_0 \leq \frac{\Bar{R}}{\Bar{ r_{+}}} + B\left(1 + \frac{\Bar{ r}_{-}}{\Bar{ r}_{+}}\right) \defeq  \Bar{l}_0 \label{eq:lbar_0}
\end{align}
Note that $\Bar{l_0}$ depends only on the region $\mathcal{U}(\mathcal{X})$.

Since $\mathcal{U}_j \subset \mathcal{U}(\mathcal{X})$, the following inequalities hold:
\begin{align} 
    \Bar{R} \geq \Bar{R}_j\text{, }~\Bar{r}_{+} \leq \Bar{r}_{+, j}\text{, }~\Bar{r}_{-} &\geq \Bar{r}_{-, j} \label{eq:ineq_Rrr}
\end{align}
\cref{eq:lbar,eq:lbar_0,eq:ineq_Rrr} leads to $\Bar{l} \leq \Bar{l}_0$.

\textbf{(ii)} $\bm{q}\in\mathcal{U}_j$ $\land$ $\textbf{GT}(\bm{q})\notin\mathcal{U}_j$ $\land$ $\Vert\bm{q} - \textbf{GT}(\bm{q})\Vert_2 \leq \Bar{R} - \Bar{R}_j$

We assume that condition (ii) in \cref{thm:upperbound_proof} is met. Since $\textbf{GT}(\bm{q})\notin\mathcal{U}_j$, $\textbf{GT}(\bm{q})\in\mathcal{U}_k$ $(j\neq k)$ holds. From the triangle inequality, the following inequality holds:
\begin{align}
    \Vert\bm{d}_j - \textbf{GT}(\bm{q})\Vert_2 \leq \Vert\bm{d}_j - \bm{q}\Vert_2 + \Vert\bm{q} - \textbf{GT}(\bm{q})\Vert_2
\end{align}
We define a constant $\Delta_{\bm{q}} \defeq  \Vert\bm{q} - \textbf{GT}(\bm{q})\Vert_2$. Since $\bm{d}_j, \bm{q}\in\mathcal{U}_j$ (\cref{fig:b-monotonic-example-cluster} (b)), the following inequality holds:
\begin{align}
    \Vert\bm{d}_j - \textbf{GT}(\bm{q})\Vert_2 \leq \Bar{R}_j + \Delta_{\bm{q}}.
\end{align}
In the same way as \textbf{(i)}, we evaluate $l$ by:
\begin{align}
    l \leq \frac{\Bar{R}_j + \Delta_{\bm{q}}}{\Bar{ r_{+}}} + B\left(1 + \frac{\Bar{ r}_{-}}{\Bar{ r}_{+}}\right) \defeq \Bar{l}. \label{eq:lbar_case2}
\end{align}
Note that we replaced $r_{+,j}, r_{-,j}$ in \cref{eq:lbar} with $r_{+}, r_{-}$ because $\textbf{GT}(\bm{q})\notin\mathcal{U}_j$ unlike \textbf{(i)}.

In contrast, we evaluate the number of hops $l_0$ in a path starting from the fixed central point $\phi(\bm{d}_0)$ as follows:
\begin{align}
    l_0 \leq \frac{\Bar{R}}{\Bar{ r_{+}}} + B\left(1 + \frac{\Bar{ r}_{-}}{\Bar{ r}_{+}}\right) \defeq \Bar{l}_0. \label{eq:lbar_0_case2}
\end{align}

The condition $\Delta_q \leq \Bar{R} - \Bar{R}_j$ leads to $\Bar{l} \leq \Bar{l}_0$.

\end{proof}


\subsection{Comparison to Previous Works}
Our theoretical analysis extends the one provided in~\cite{ni2023diskann++}. The previous work~\cite{ni2023diskann++} showed that the adaptive entry point selection provides a better upper bound of hops than the fixed central entry point, assuming the graph is MSNET~\cite{dearholt1988monotonic_MSNET}. Considering those outcomes, we introduced the generalized concept $B$-MSNET (\cref{def:B-MSNET}), which does not require that a pair of nodes have an exact \textit{monotonic} path. It better captures the actual graph in practical algorithms~\cite{Fu2017FastAN_NSG,jayaram2019diskann}. Then, we proved that the same statement holds even when the graph is $B$-MSNET (\cref{thm:upperbound_proof}). In addition, the existing theoretical statements require that the data be distributed in a unit sphere. We loosen this assumption in \cref{thm:upperbound_proof} by considering the finite region $\mathcal{U}(\mathcal{X})$ and its Voronoi partition $\{\mathcal{U}_1,\dots,\mathcal{U}_K\}$. Finally, we simplify the existing approach by introducing \cref{lemma:b-monotonic-sum-of-r}. It clarifies things even when considering a more complicated concept, $B$-MSNET.

\section{Empirical Findings}
\subsection{Experiment Settings}
We evaluate the algorithm regarding accuracy, speed, and memory usage. We measure the accuracy of the algorithm by $\text{Recall@}k$. Given the ground-truth top $k$ neighbor of a query $\mathcal{R}$ and the output top $k$ neighbor $\hat{\mathcal{R}}$, we define $\text{Recall@}k = \frac{\vert\mathcal{R}\cap\hat{\mathcal{R}}\vert}{k}$. We evaluate the speed of the algorithm by QueriesPerSecond (QPS), which corresponds to the throughput of the algorithm. We define it as the average number of processed queries per second.

We conducted experiments on an Intel(R) Core(TM) i7-10870H CPU @ 2.20GHz with 64GB RAM, setting the number of threads to 8. For NSG~\cite{Fu2017FastAN_NSG} and k-means clustering, we used the implementation provided by the \texttt{Faiss} library~\cite{johnson2019billion_faiss,douze2024faiss_library}. We set $R=32, L=64, \text{and}~C=132$ for NSG. The implementation of NSG we adopt uses NN-Descent~\cite{dong2011efficient_Kgraph_NNDescent} as the base KNN graph. The parameters of NN-Descent are $K=64, L=114, R=100, S=10, \text{and } iter=10$. They are the default settings. For DiskANN~\cite{jayaram2019diskann}, we used the official implementation~\cite{diskann-github} by Microsoft. We set $R=70, L=125, \text{and}~\alpha=1.2$ for DiskANN.

\subsection{Evaluation on various datasets} \label{sec:eval_on_various_dataset}
We demonstrate the adaptively selected entry point outperforms the fixed central entry point on NSG index. We used 8 datasets with various data characteristics. \cref{table:datasets} describes the dimensionality and the size of query set. \textbf{SIFT 1M} and \textbf{GIST 1M}~\cite{Jegou2011ProductQuantization} consists of classical image descriptors. \textbf{Deep1M}~\cite{babenko2016efficientDeep1M}, \textbf{OpenAI 1M}~\cite{MS2023_BigANNBenchRepo}, and \textbf{CLIP I2I 1M} are feature vectors from deep neural networks. OpenAI 1M consists of text embeddings from wiki dataset. CLIP I2I 1M consists of CLIP~\cite{radford2021learningCLIP} image embeddings from the LAION 5B dataset~\cite{schuhmann2022laion5b}. \textbf{Gauss 1M} consists of artificial samples from a gaussian mixture distribution with 10 mixture components. \textbf{Yandex T2I 1M}~\cite{Simhadri2022NeuripsComp} and \textbf{CLIP T2I 1M} are OOD datasets, where the statistical distribution of the database and queries are different. CLIP T2I 1M consists of image embeddings from the LAION 400M dataset~\cite{schuhmann2021laion} as the database and text embeddings as queries.



\begin{table}[tb]
\caption{Dataset descriptions. `Dim` means the dimensionality of vectors.}
\label{table:datasets}
\vskip 0.15in
\begin{center}
\begin{small}
\begin{tabular}{@{}lccl@{}}
\toprule
Dataset & Dim & \#Query & Description \\
\midrule
SIFT 1M & 128 & 10,000 & SIFT Descriptor \\
GIST 1M & 960 & 1,000 & GIST Descriptor \\
Deep 1M & 96 & 10,000 & Image Embedding \\
OpenAI 1M & 1536 & 10,000 & Text Embedding\\
CLIP I2I 1M & 512 & 10,000 & Image Emebdding \\
Gauss 1M & 128 & 10,000 & Random Samples\\
Yandex T2I 1M & 200 & 100,000 & Image \& Text Embedding\\
CLIP T2I 1M & 768 & 10,000 & Image \& Text Embedding\\
\bottomrule
\end{tabular}
\end{small}
\end{center}
\vskip -0.1in
\end{table}

\begin{figure*}[tb]
\centering
\vskip 0.2in

\begin{center}
\begin{tabular}{@{}cccc@{}}
\subf{\includegraphics[width=0.23\textwidth]{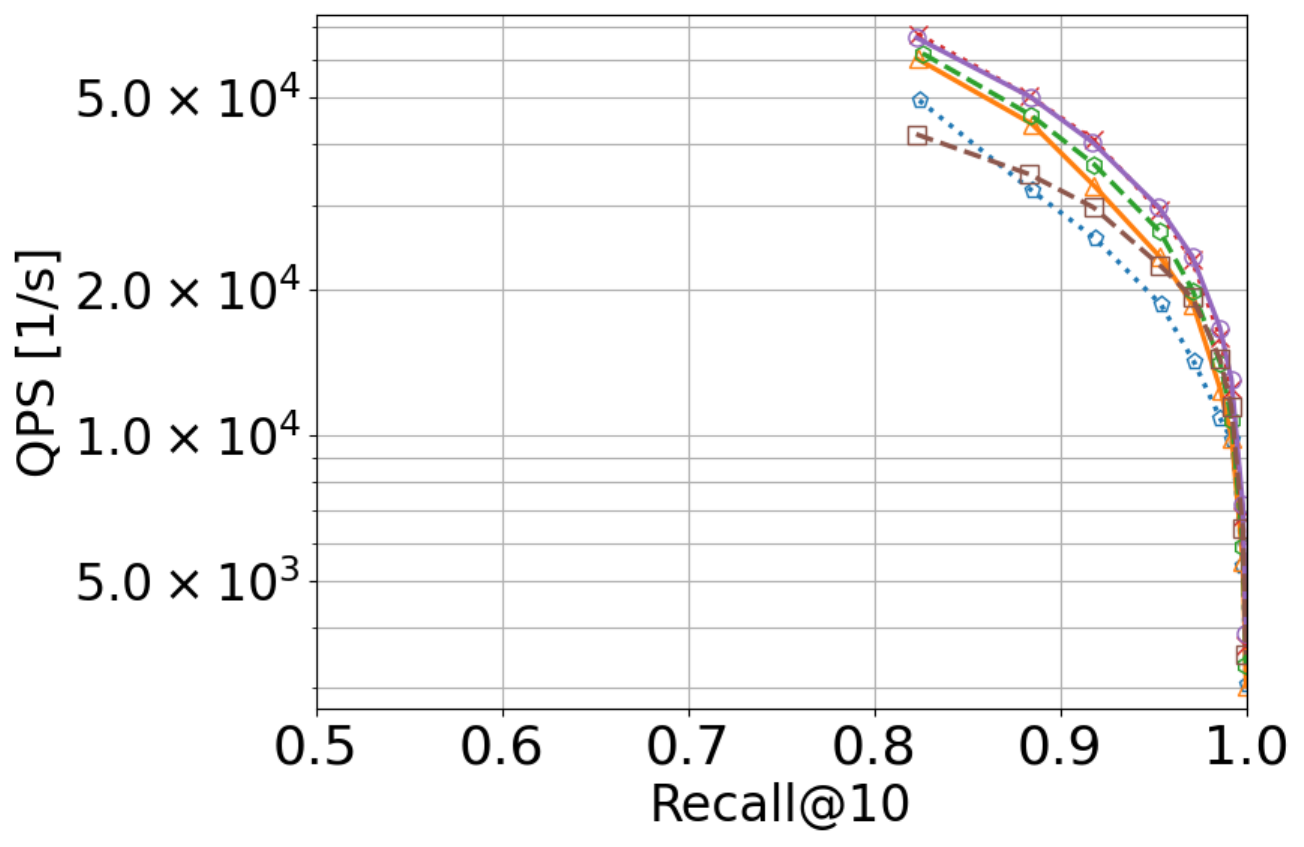}}
     {(a) SIFT 1M}
&
\subf{\includegraphics[width=0.23\textwidth]{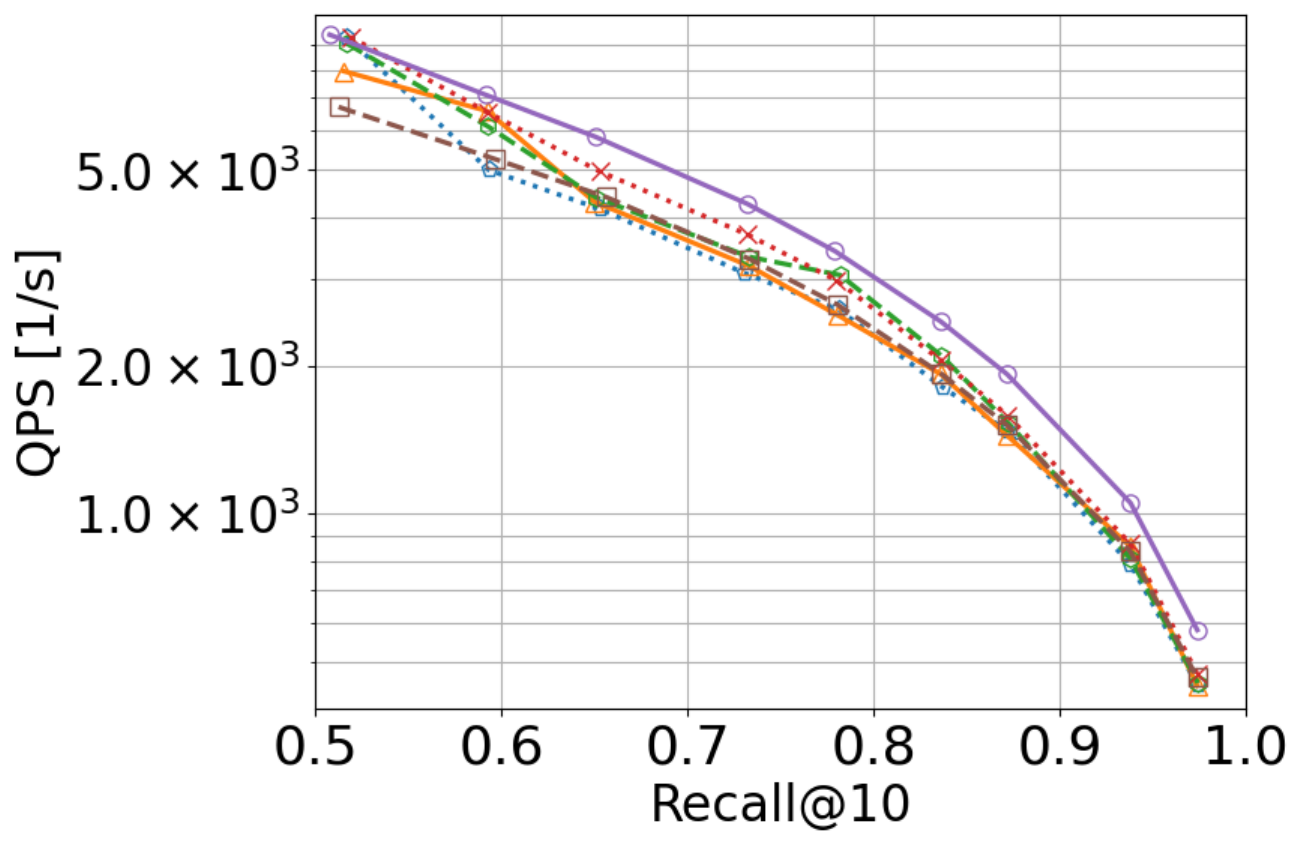}}
     {(b) GIST 1M}
&
\subf{\includegraphics[width=0.23\textwidth]{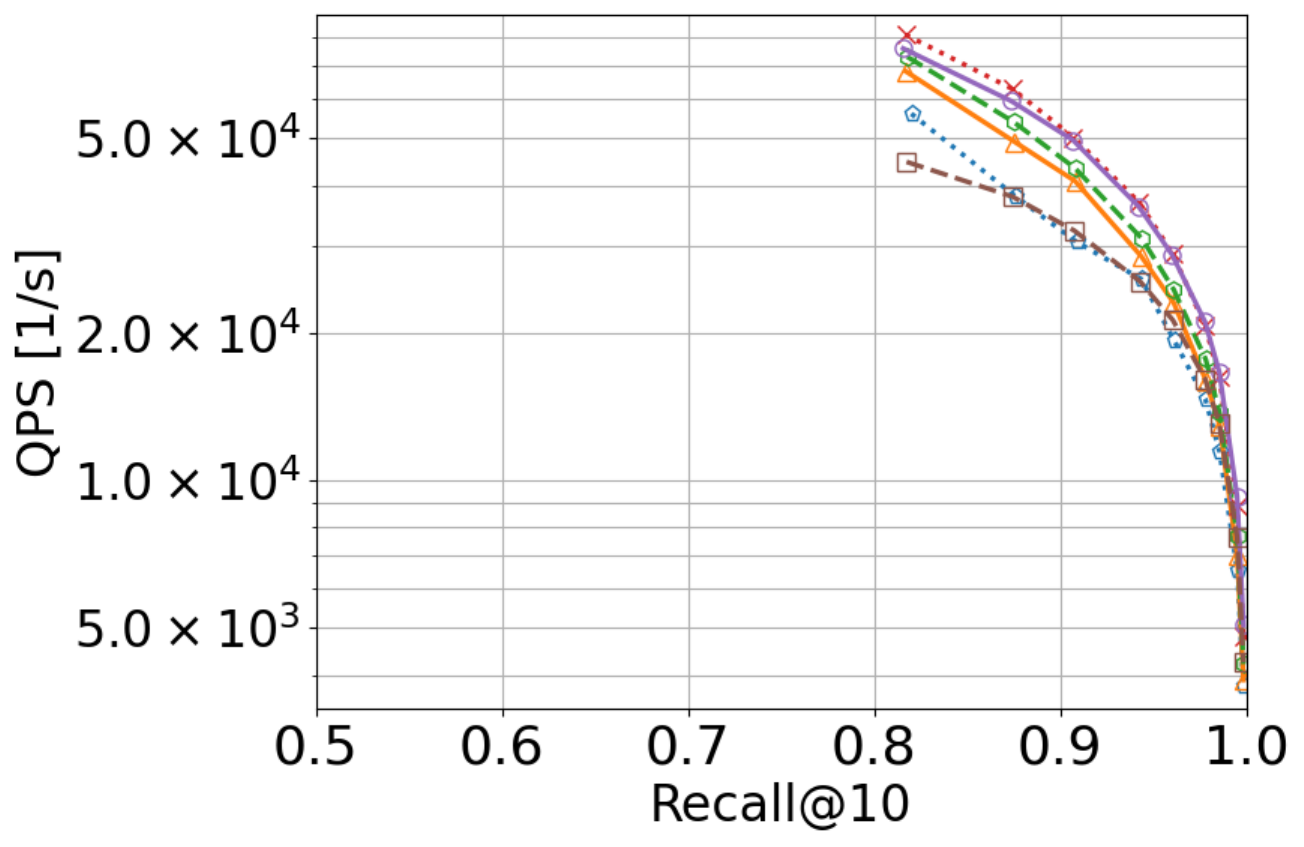}}
     {(c) Deep 1M}
&
\subf{\includegraphics[width=0.23\textwidth]{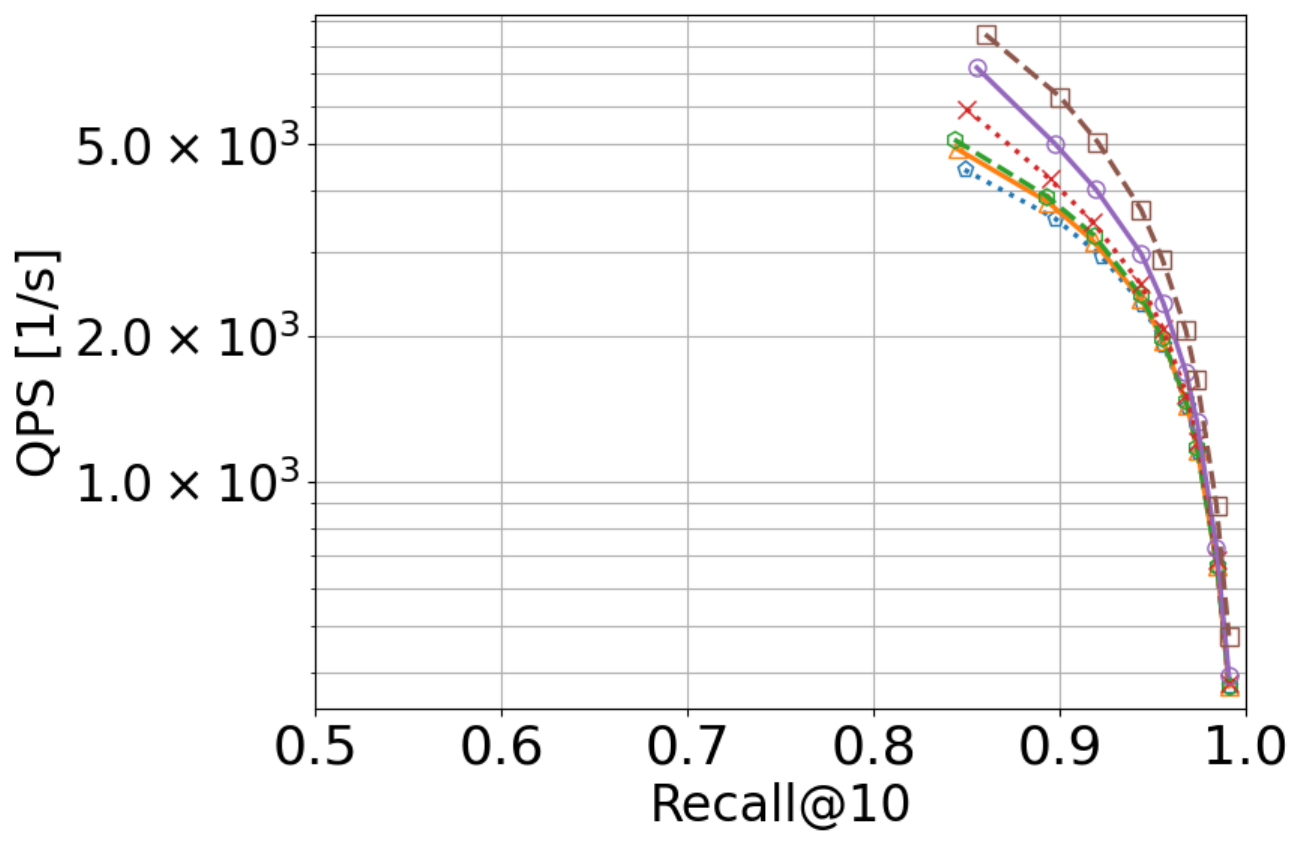}}
     {(d) OpenAI 1M}
\\
\subf{\includegraphics[width=0.23\textwidth]{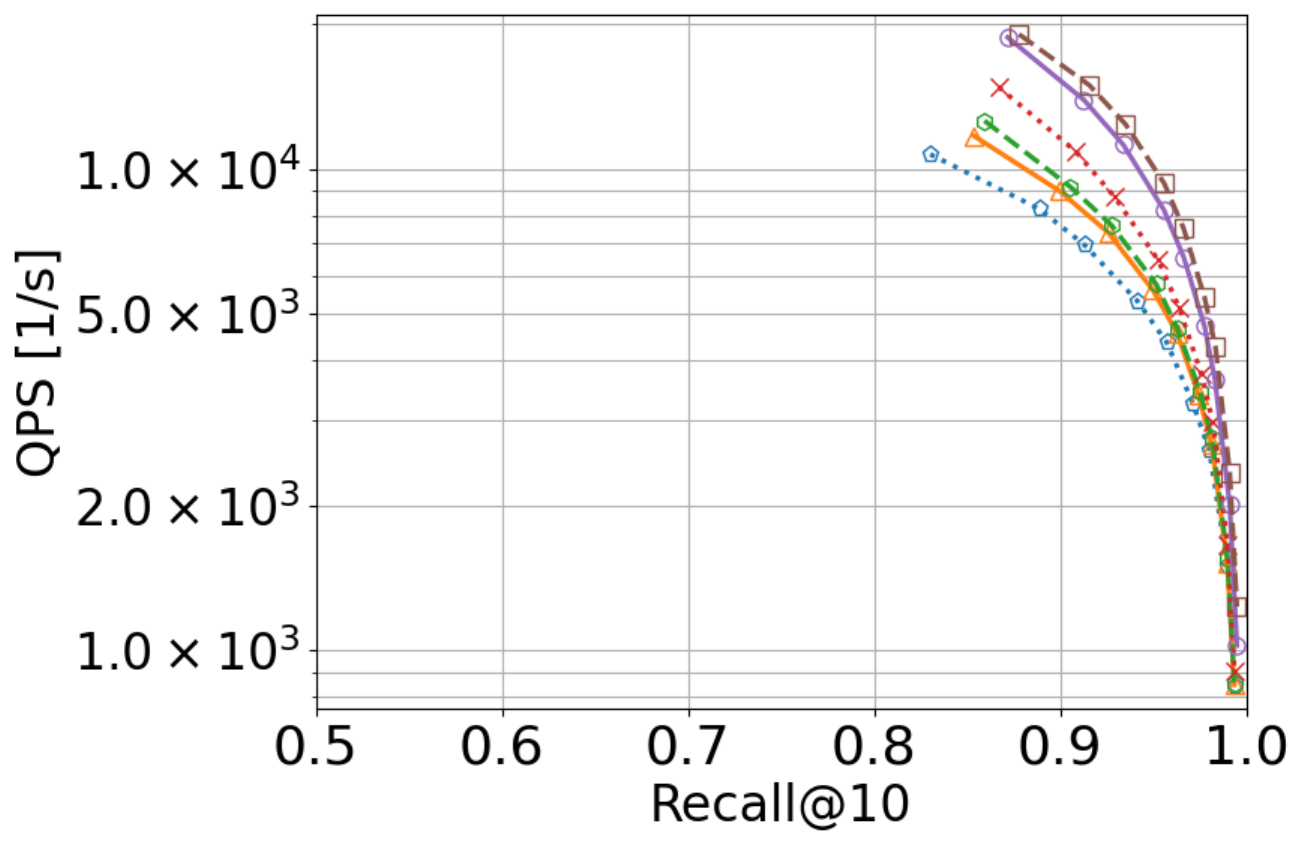}}
     {(e) CLIP I2I 1M}
&
\subf{\includegraphics[width=0.23\textwidth]{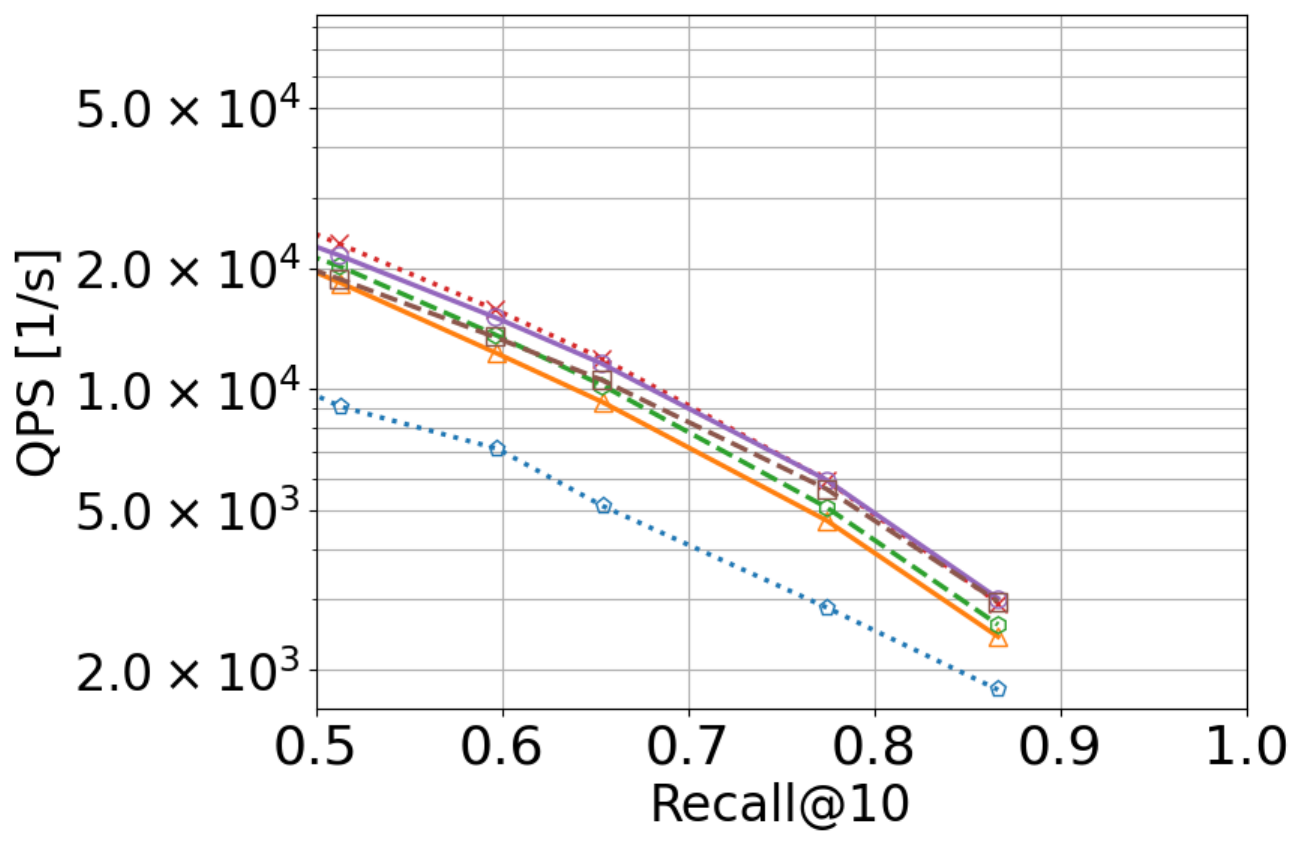}}
     {(f) Gauss 1M}
&
\subf{\includegraphics[width=0.23\textwidth]{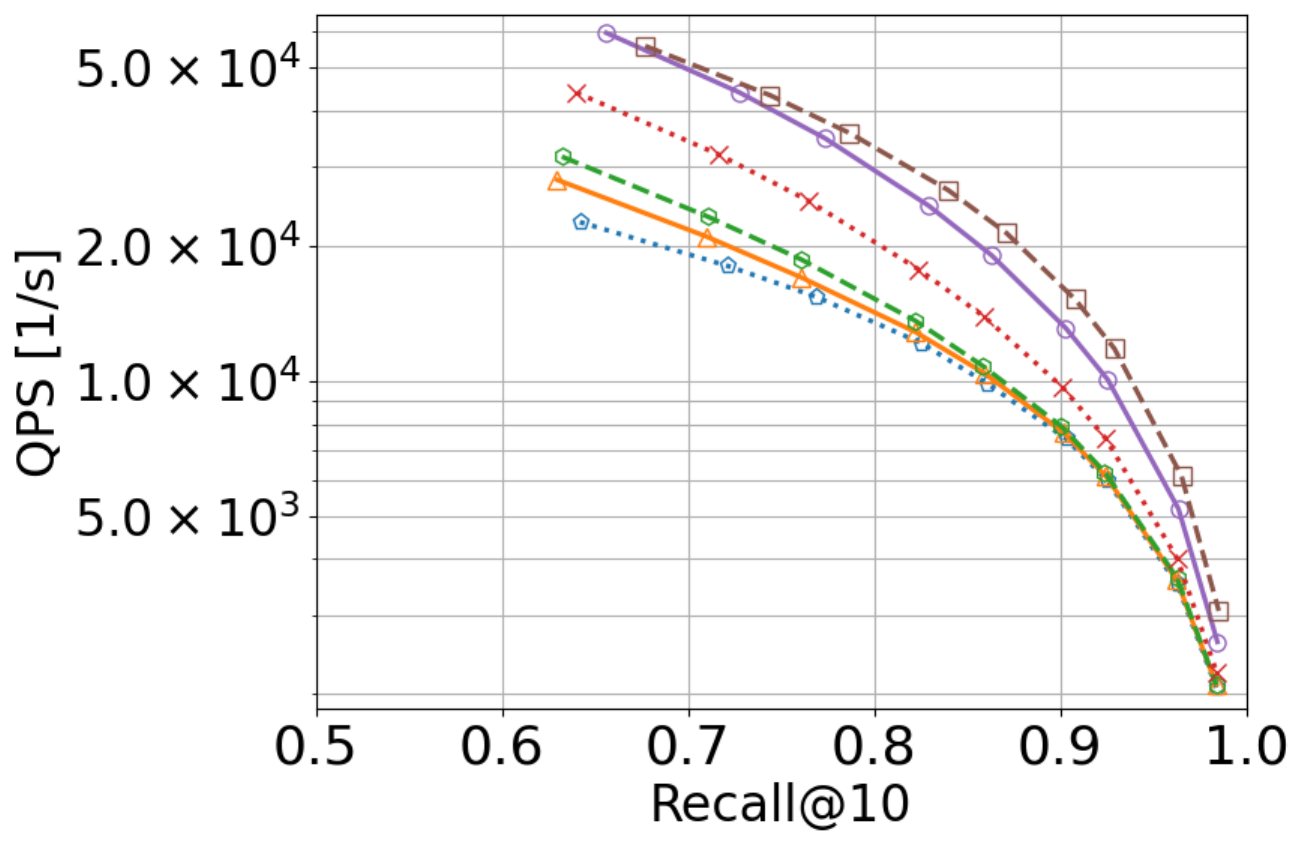}}
     {(g) Yandex T2I 1M}
&
\subf{\includegraphics[width=0.23\textwidth]{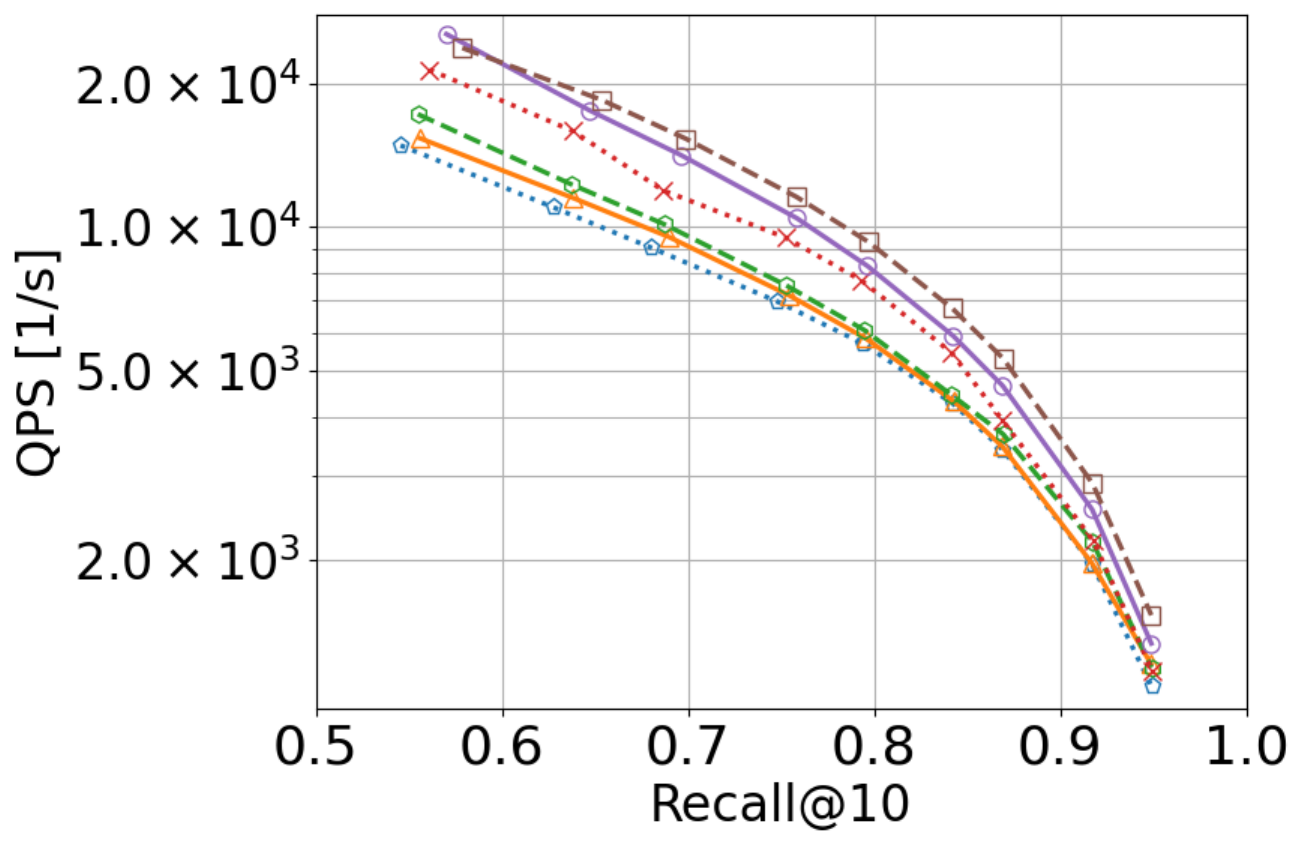}}
     {(h) CLIP T2I 1M}
\\
\end{tabular}
\end{center}

\begin{subfigure}
    \centering
    \includegraphics[width=0.7\textwidth]{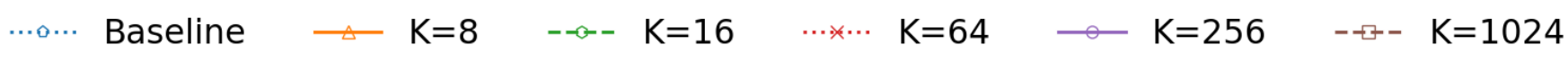}
\end{subfigure}

\caption{The evaluation of NSG~\cite{Fu2017FastAN_NSG} with adaptive entry point selection on various datasets. A curve on the upper right side is better than others in terms of accuracy-speed tradeoff. We sweep a curve by changing the length of the search queue $L\in\{16, 24, 32, 48, 64, 96, 128, 256, 512\}$. We took the average of five measurements for each cases.}
\label{fig:result_nsg}
\vskip -0.2in
\end{figure*}

\cref{fig:result_nsg} shows the result when we apply the entry point selection method to NSG~\cite{Fu2017FastAN_NSG} index. We regard the vanilla NSG as the baseline and compare it to the cases on different $K$. We observe that the QPS improves 1.2 - 2.3 times for all datasets. \cref{fig:result_nsg} (a) - (d) shows the effectiveness in standard datasets with different dimensionalities. We also see the Recall@10 exhibits not so small improvement by 0.01 - 0.04 in CLIP I2I 1M, CLIP T2I 1M, and Yandex T2I 1M (\cref{fig:result_nsg} (e), (g), (h)). Remarkably, the QPS largely improved by 2.3x on a clustered dataset like Gauss 1M (\cref{fig:result_nsg} (f)).

We also demonstrate that the adaptive entry point selection has only a tiny amount of overhead. \cref{table:memory_usage} shows the memory overhead and the preparation time. We chose the record with the best tradeoff among the ones shown in \cref{fig:result_nsg}. It illustrates that we require only less than roughly 0.1\% of the original index size to obtain the best performance by the entry point selection. Additionally, the preparation of candidates is fast and poses almost no practical problem.

\begin{table}[tb]
\caption{The overhead memory usage (Mem. overhead) and the preparation time (Prep. time) of the adaptive entry point selection. The overhead is a ratio of the size of the additional index over that of the original index when achieving the best tradeoff among $K\in\{1, 8, 16, 64, 256, 1024\}$.}
\label{table:memory_usage}
\vskip 0.15in
\begin{center}
\begin{small}
\begin{tabular}{@{}lccc@{}}
\toprule
Dataset & Mem. overhead & Prep. time [sec] & $K$ \\
\midrule
SIFT 1M & 0.0055\% & 0.62 & $64$ \\
GIST 1M & 0.0252\% & 6.46 & $256$ \\
Deep 1M & 0.0052\% & 0.49 & $64$ \\
OpenAI 1M & 0.102\% & 46.87 & $1024$ \\
CLIP I2I 1M & 0.101\% & 22.67 & $1024$ \\
Gauss 1M & 0.0056\% & 4.16 & $64$ \\
Yandex T2I 1M & 0.0992\% & 6.00 & $1024$ \\
CLIP T2I 1M & 0.099\% & 16.14 & $1024$ \\
\bottomrule
\end{tabular}
\end{small}
\end{center}
\vskip -0.1in
\end{table}

\subsection{Overcoming Hard Instances} \label{sec:hardinstance}
In this section, we recap hard instances for graph-based index presented in~\cite{indyk2023worstcase}. The accuracy of existing graph-based indexes significantly drops in the instances. The adaptive entry point selection can overcome such hard instances by achieving non-zero accuracy with a much faster search than the vanilla index.


\begin{figure}[tb]
\vskip 0.2in
\centering
    \subfigure[NSG]{%
        \includegraphics[clip, width=0.46\columnwidth]{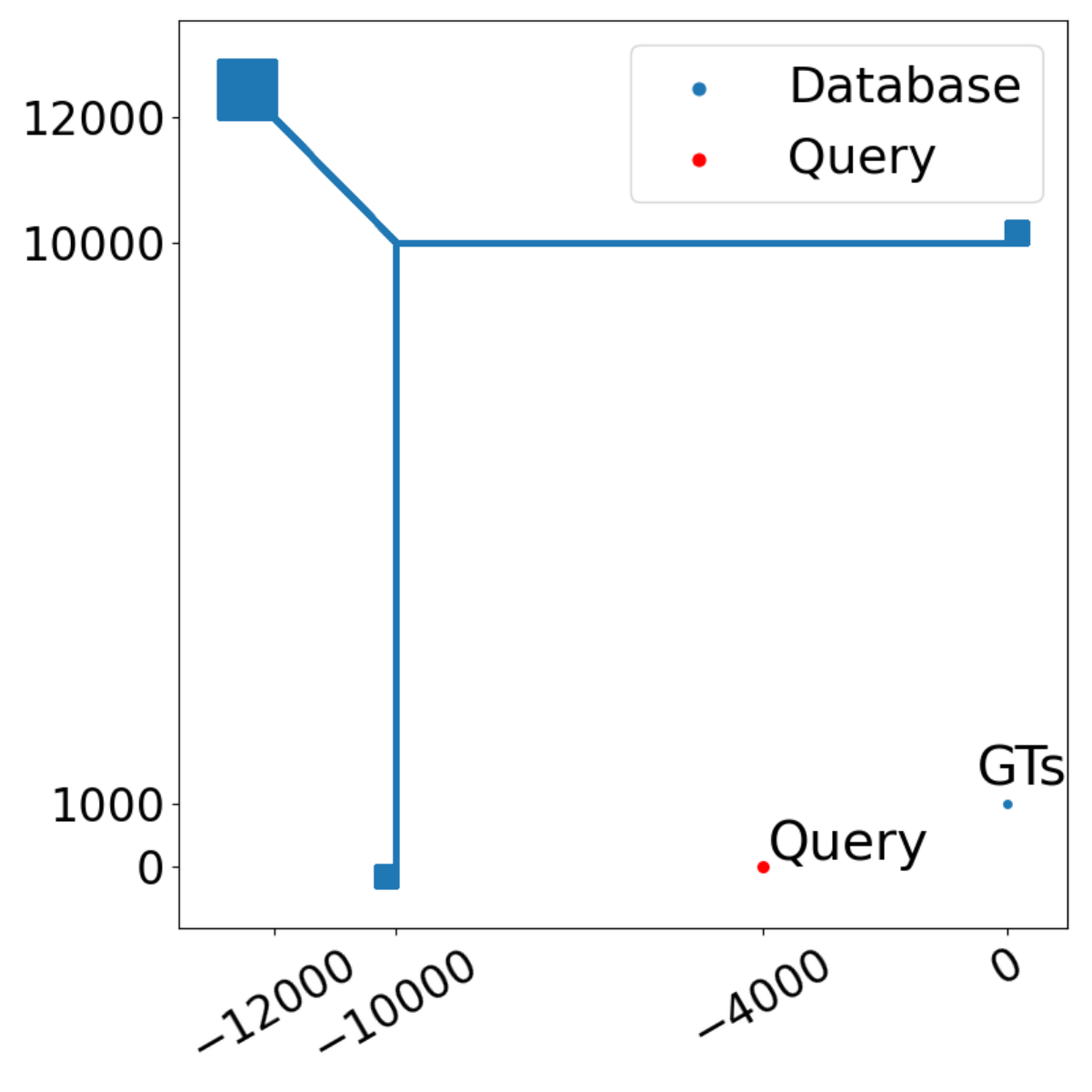}}%
    \subfigure[DiskANN]{%
        \includegraphics[clip, width=0.46\columnwidth]{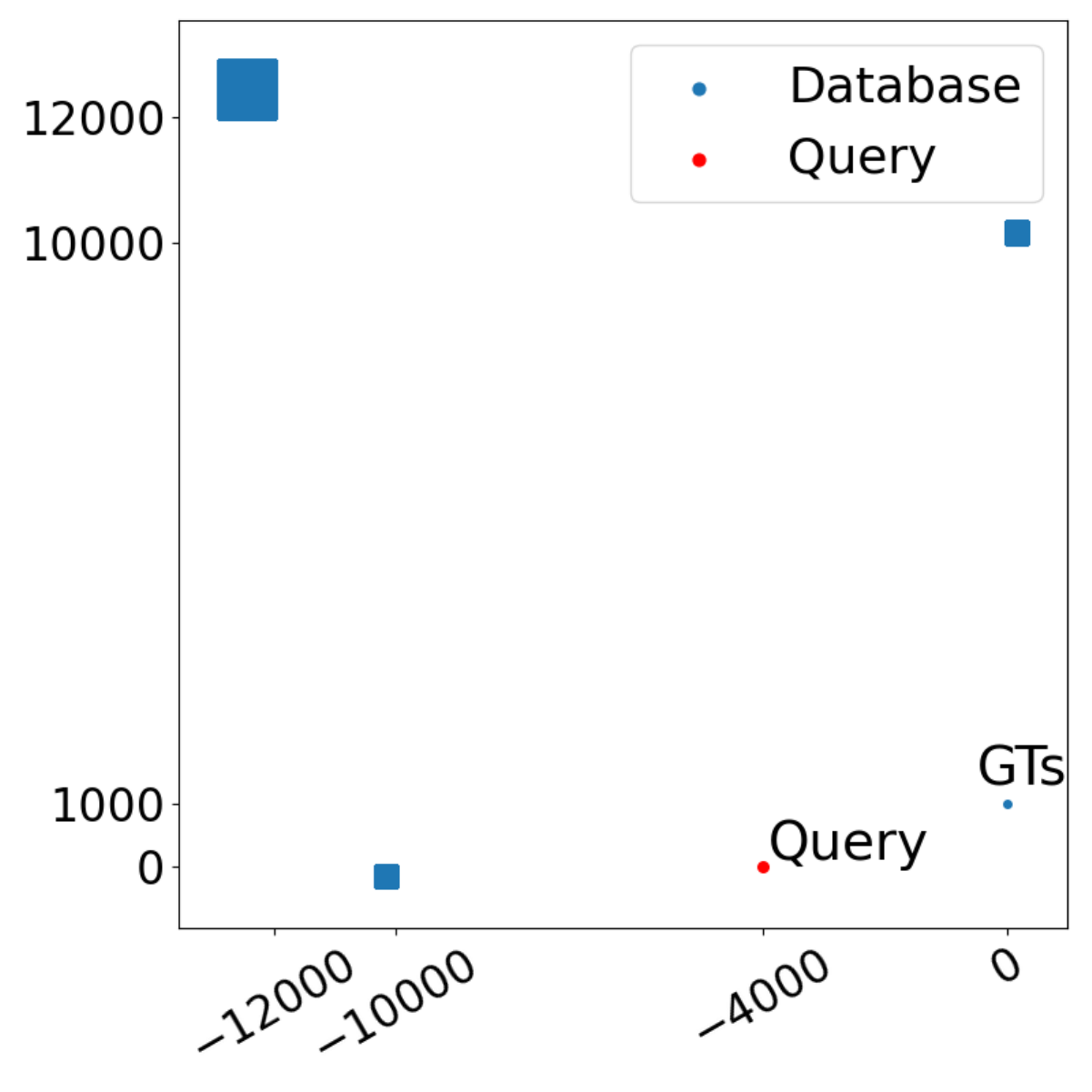}}%
    \caption{Visualization of reproduced hard instances presented in~\cite{indyk2023worstcase} for (a) NSG and (b) DiskANN. GTs represent the ground truth points.}
    \label{fig:worst_vis}
\vskip -0.2in
\end{figure}

Let us first recap the hard instances. We reproduced the hard case instances with 1M samples as shown in \cref{fig:worst_vis}. Because we evaluate all experiments with $\text{Recall}@10$, we create a tiny cluster of $10$ samples as ground truth samples. Their positions are the same as proposed in the original instance~\cite{indyk2023worstcase}.

\begin{figure}[tb]
\vskip 0.2in
\centering
    \subfigure[NSG]{%
        \includegraphics[clip, width=0.5\columnwidth]{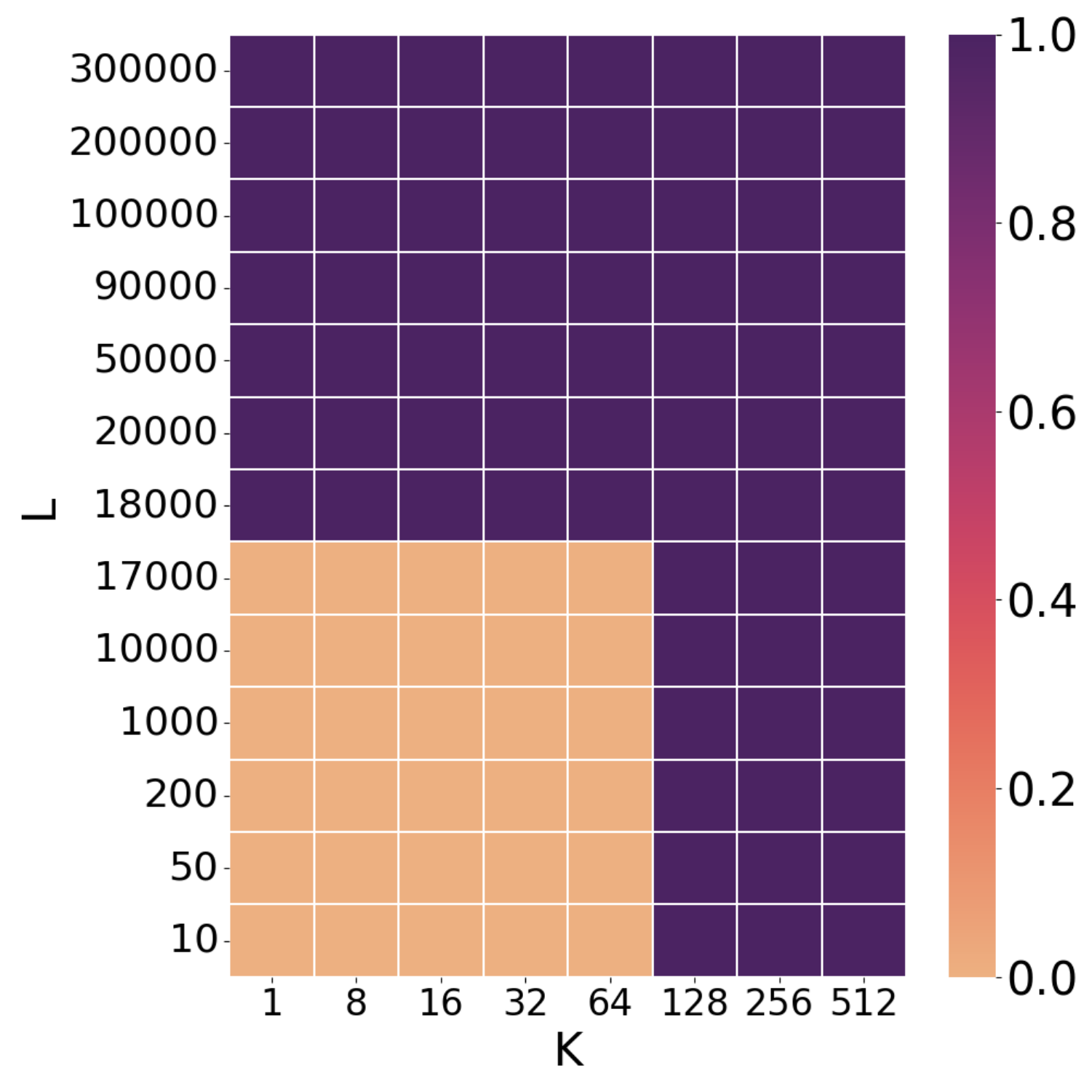}}%
    \subfigure[DiskANN]{%
        \includegraphics[clip, width=0.5\columnwidth]{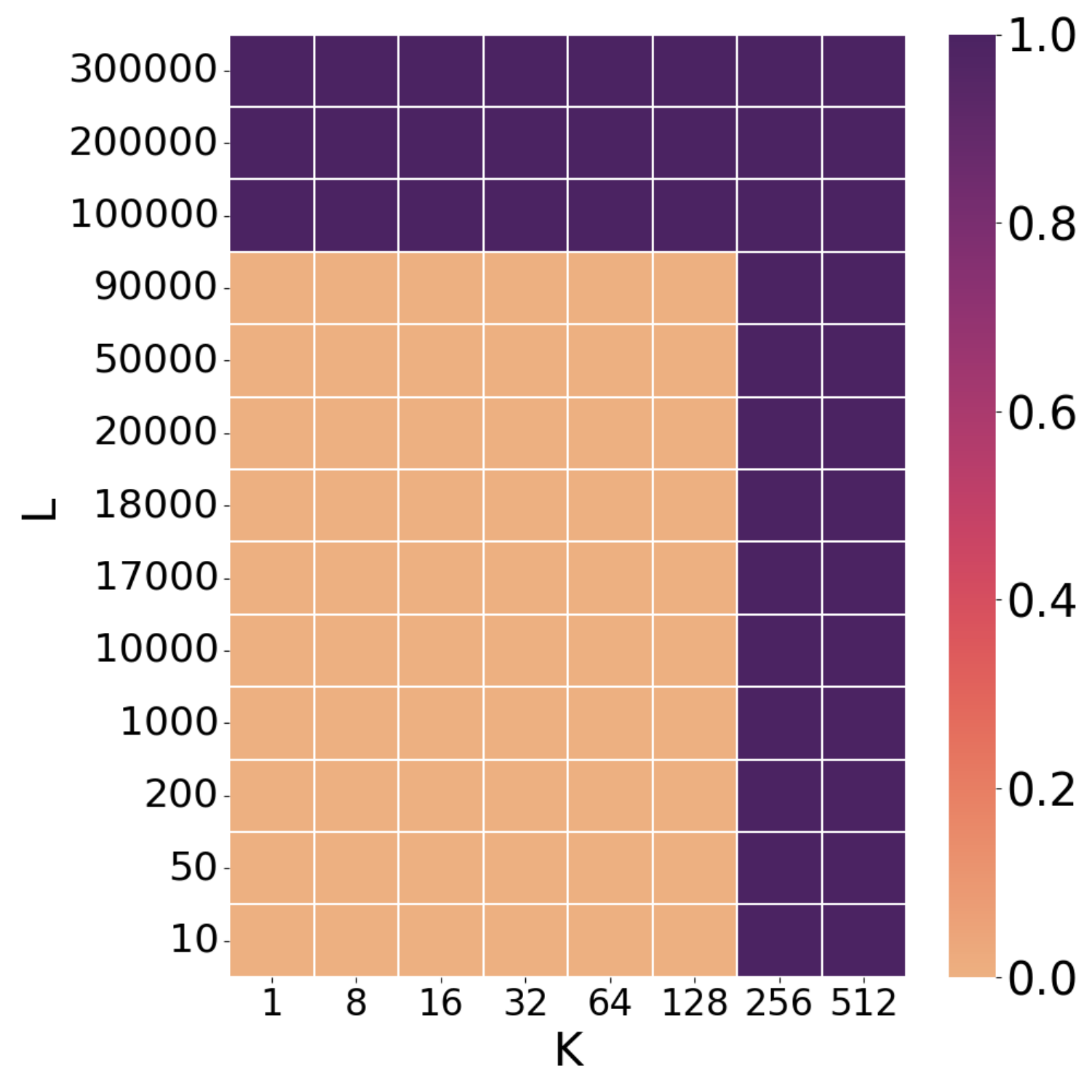}}%
    \caption{Heatmaps that represent the results of (a) NSG and (b) DiskANN on the hard instances when varying number of entry point candidates $K$ and the length of search queue $L$. Each cell in the heatmap represents the $\text{Recall@}10$. It shows only a part of all $L$ for better visualization.}
    \label{fig:worst_heatmap}
\vskip -0.2in
\end{figure}

\begin{figure}[tb]
\vskip 0.2in
\centering
    \subfigure[$K=16$ (Failed)]{%
        \includegraphics[clip, width=0.48\columnwidth]{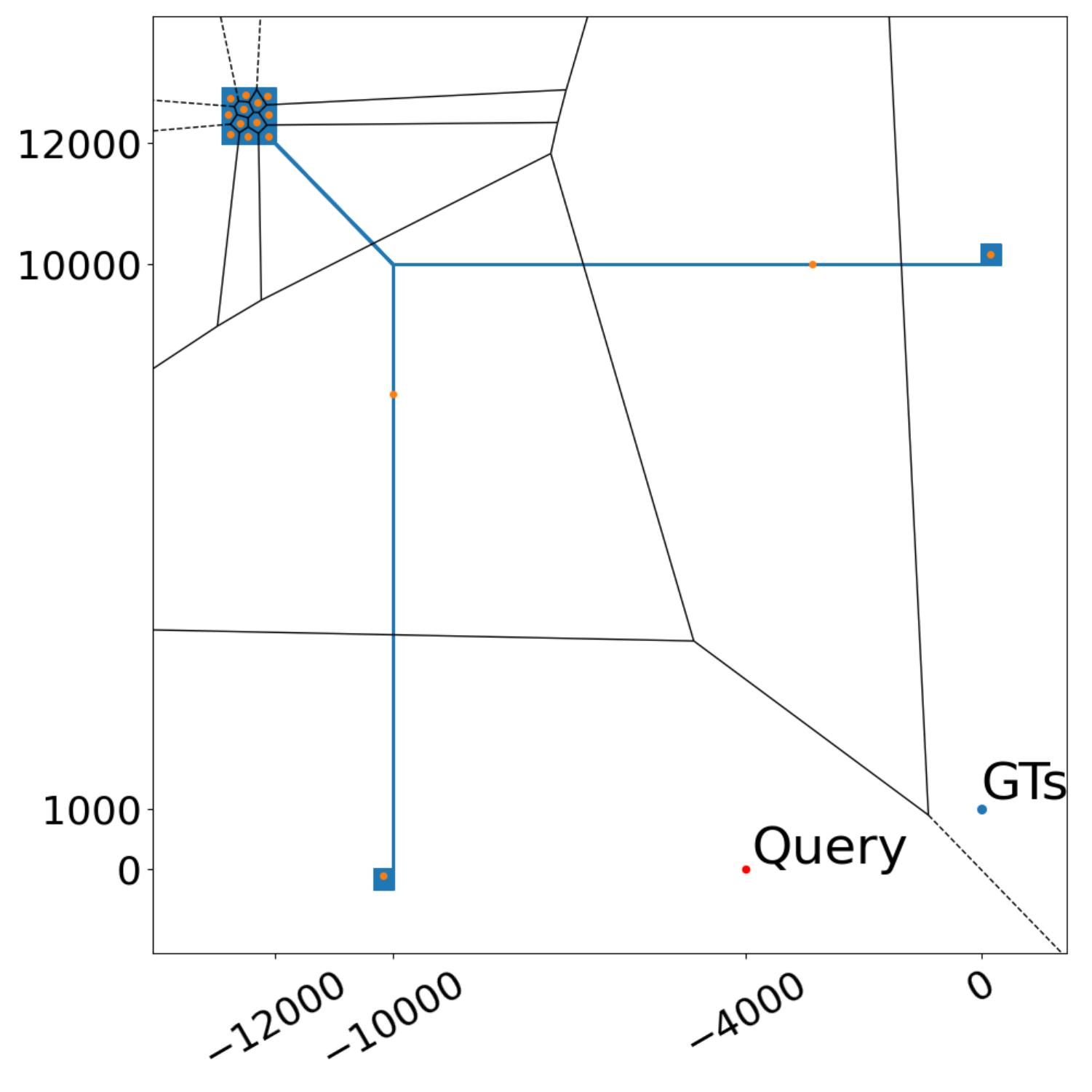}}%
    \subfigure[$K=128$ (Succeeded)]{%
        \includegraphics[clip, width=0.48\columnwidth]{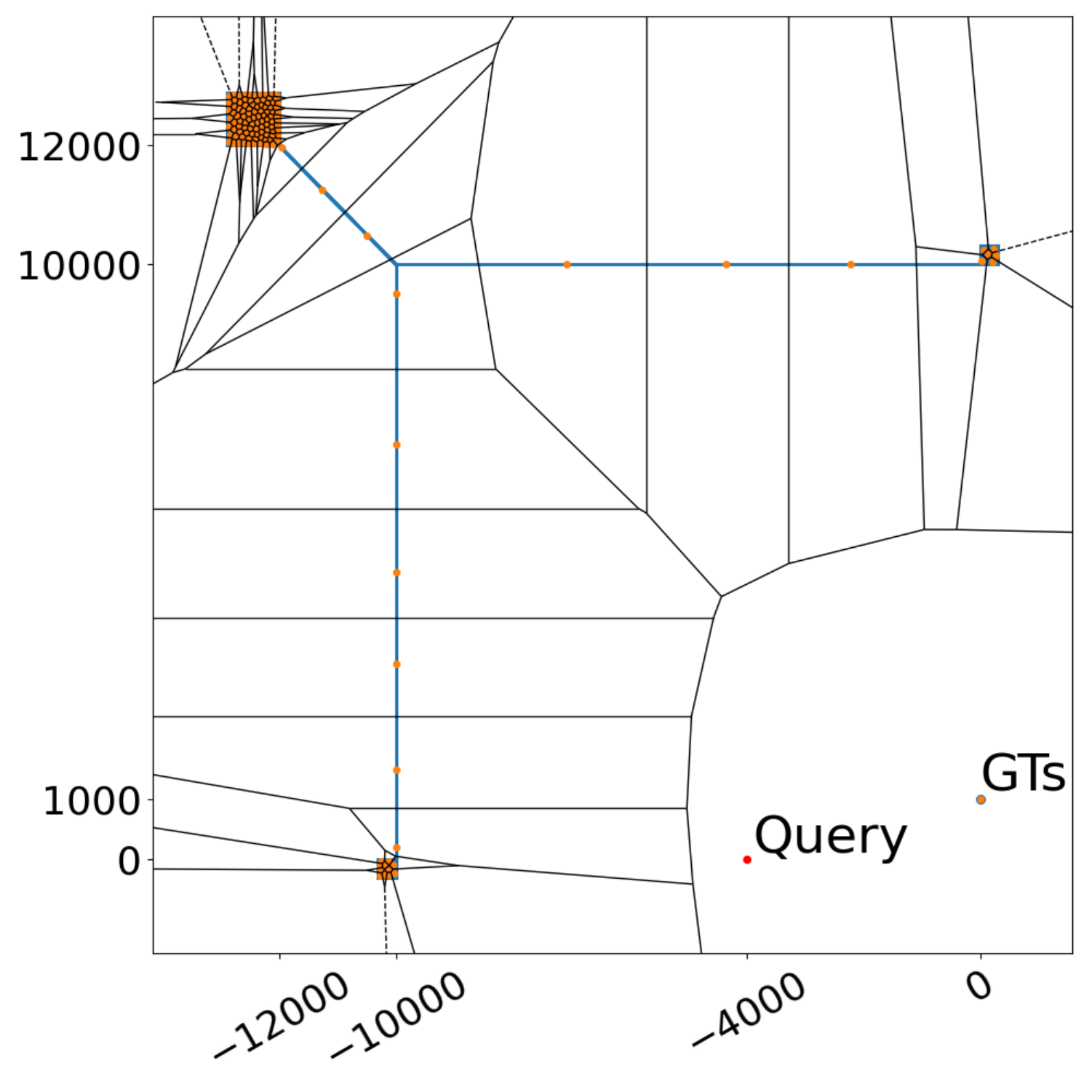}}%
    \caption{The Voronoi partition on the hard instance for NSG. (a) $K=16$ leads to zero accuracy in search (failure), but (b) $K=128$ leads to Recall@$10$ = 1.0 (success). Blue points show samples in the database $\mathcal{X}$. Orange points represent entry point candidates $\mathcal{D}$.}
    \label{fig:worst_voronoi}
\vskip -0.2in
\end{figure}

We select an entry point adaptively (\cref{sec:recap_ep_selection}) and perform a search on the hard instances. The target indexes are NSG and DiskANN. We change the number of entry point candidates $K\in\{1, 8, 16, 32, 64, 128, 256, 512\}$, where $K=1$ means the vanilla index without applying the adaptive selection. We also change the length of the search queue $L$. We sample appropriate $L$ ranging from $10$ to $1,000$ by $10$, from $1,000$ to $20,000$ by $1,000$, and from $20,000$ to $300,000$ by $10,000$. Considering the dataset size is 1M, the minimum value $10$ is reasonably small, and the maximum value $300,000$ is significantly large.

\begin{table*}[tb]
\caption{Improvement of the best QPS to reach non-zero accuracy for NSG and DiskANN on hard instances}
\label{table:worst_qps}
\vskip 0.15in
\begin{center}
\begin{small}
\begin{tabular}{@{}lccc@{}}
\toprule
Index & QPS (vanilla) & QPS (with entry point selection) & Improvement \\
\midrule
NSG & 30.98 $(K=1, L=18,000)$ & 10754 $(K=128, L=10)$ & \textbf{347x} \\
DiskANN & 28.92 $(K=1, L=100,000)$ & 10810 $(K=256, L=50)$ & \textbf{373x} \\
\bottomrule
\end{tabular}
\end{small}
\end{center}
\vskip -0.1in
\end{table*}

\cref{fig:worst_heatmap} shows the results of NSG and DiskANN on the hard instances. The baseline corresponds to the column of $K=1$. It shows that the baseline requires significantly large $L$ to gain non-zero accuracy. For example, NSG needs $L \geq 18,000$, and DiskANN needs $L \geq 100,000$. In contrast, given a larger $K$ value than 128 (256), we can achieve non-zero accuracy in NSG (DiskANN) even when $L$ is pretty small. Therefore, the adaptive entry point selection method can help the index to overcome the hard instances. Regarding efficiency, as shown in \cref{table:worst_qps}, the adaptive entry point selection significantly improves the maximum QPS to reach non-zero accuracy. The efficiency improved \textbf{347} times in NSG and \textbf{373} times in DiskANN.


\cref{fig:worst_voronoi} shows the Voronoi partitions for NSG on the hard instances. The representative points of the partitions are entry point candidates $\mathcal{D}$. The entry point candidates gather around three islands with many database points. As $K$ increases, entry point candidates are finally located on a distant small island where the ground truth is, which is annotated with the label "GTs" in \cref{fig:worst_voronoi}. That enables us to reach the ground truths immediately.

\subsection{Parameter Sensitivity for $K$}

\begin{figure}[tb]
\vskip 0.2in
\centering
    \subfigure[$\text{Recall@}10$]{%
        \includegraphics[clip, width=0.5\columnwidth]{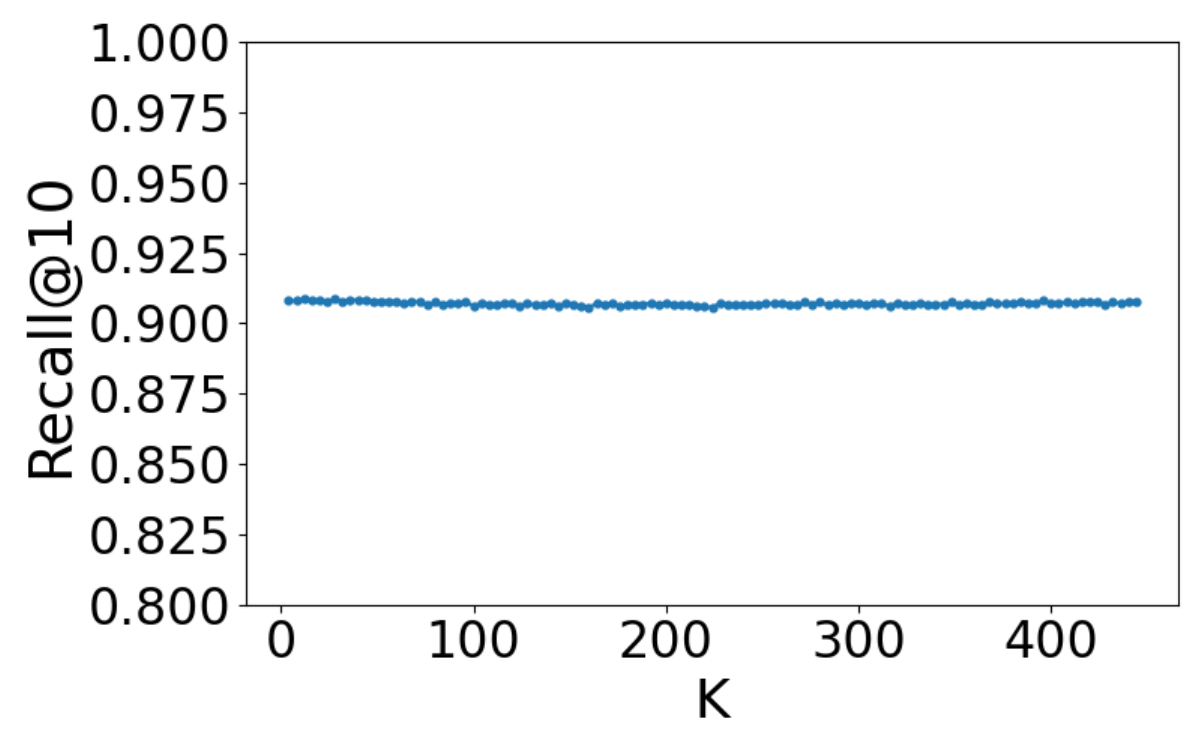}}%
    \subfigure[QPS]{%
        \includegraphics[clip, width=0.5\columnwidth]{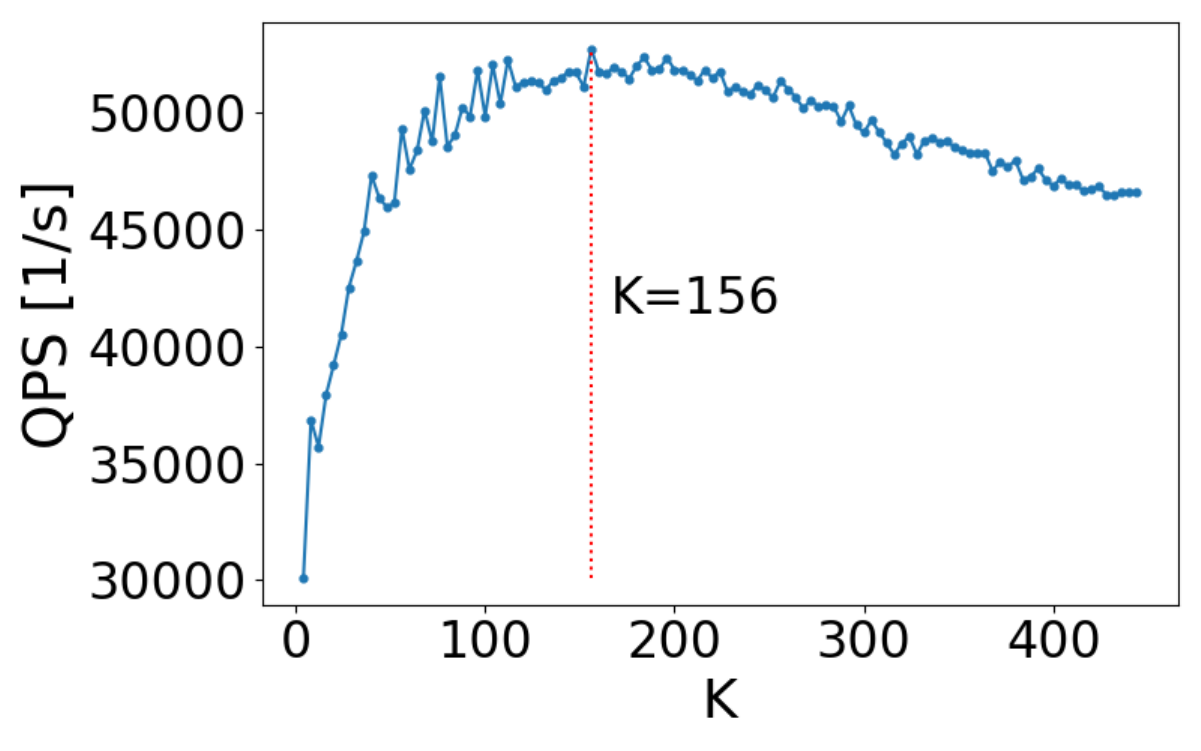}}%
    \caption{(a) $\text{Recall@}10$ and (b) QPS of NSG on Deep-1M dataset with different $K$ values. We set the length of search queue to $32$. The red dotted line in (b) shows a peak point of QPS-$K$ curve, which corresponds to $K=156$.}
    \label{fig:ablation_K}
\vskip -0.2in
\end{figure}

We studied the sensitivity of the performance on the number of entry point candidates $K$. We measured the accuracy and speed of NSG on Deep 1M dataset~\cite{babenko2016efficientDeep1M} for various $K\in\{4, 8, \dots, 444\}$. We chose this dataset because it exhibits the standard performance curve as demonstrated in \cref{sec:eval_on_various_dataset}.

\cref{fig:ablation_K} shows the performance curve when changing $K$. Recall@$10$ remains almost unchanged at around $0.90$, regardless of $K$. On the other hand, the QPS exhibits a roughly unimodal change, though it is somewhat jagged, with the value of $K$. These results show that the method is quite tractable. We can find the almost optimal $K$ regarding an accuracy-speed tradeoff. In this case, it is $K=156$.


\section{Conclusion}
Our study provides the theoretical and empirical analysis of the adaptive entry point selection based on k-means clustering. We introduced novel concepts of $b$\textit{-monotonic path} and $B$\textit{-MSNET}. They capture the actual graph-based indexes better than existing concepts. Our core theorem extends the previous work and demonstrates that the method is beneficial in more general situations. We also demonstrated that the adaptive entry point selection improves the search speed by 1.2 - 2.3 times with only slight memory overhead on various datasets. Remarkably, we can overcome the artificial hard instances by the entry point selection. 

Our newly introduced concepts will lead to future theoretical research on a graph-based index based on a more realistic situation. However, it is still unclear how many actual cases achieve the conditions we provide in \cref{thm:upperbound_proof}. Thus, a natural future direction would be considering such connections between the theory and the empirical findings. Another research direction would be to theoretically analyze the average performance improvement obtained by the adaptive entry point selection.

\section*{Impact Statements}
This paper presents work whose goal is to advance the field of Machine Learning. There are many potential societal consequences of our work, none which we feel must be specifically highlighted here.




\nocite{langley00}

\bibliography{example_paper}
\bibliographystyle{icml2024}

\newpage
\appendix
\onecolumn



\section{Voronoi Partitions on the Hard Instances}
This section provides visualized examples of Voronoi partitions on the hard instances~\cite{indyk2023worstcase} for NSG~\cite{Fu2017FastAN_NSG} and DiskANN~\cite{jayaram2019diskann}. The ones in \cref{sec:hardinstance} correspond to the case of $K=16$ and $K=128$ for NSG. We list the other examples in \cref{fig:appendix_worst_voronoi_1}. They demonstrate that the query and the ground truths are in the different Voronoi cells in failed cases but are in the same cells in succeeded cases.

\begin{figure}[tb]
\vskip 0.2in
\centering
    \subfigure[NSG: $K=64$ (Failed)]{%
        \includegraphics[clip, width=0.33\columnwidth]{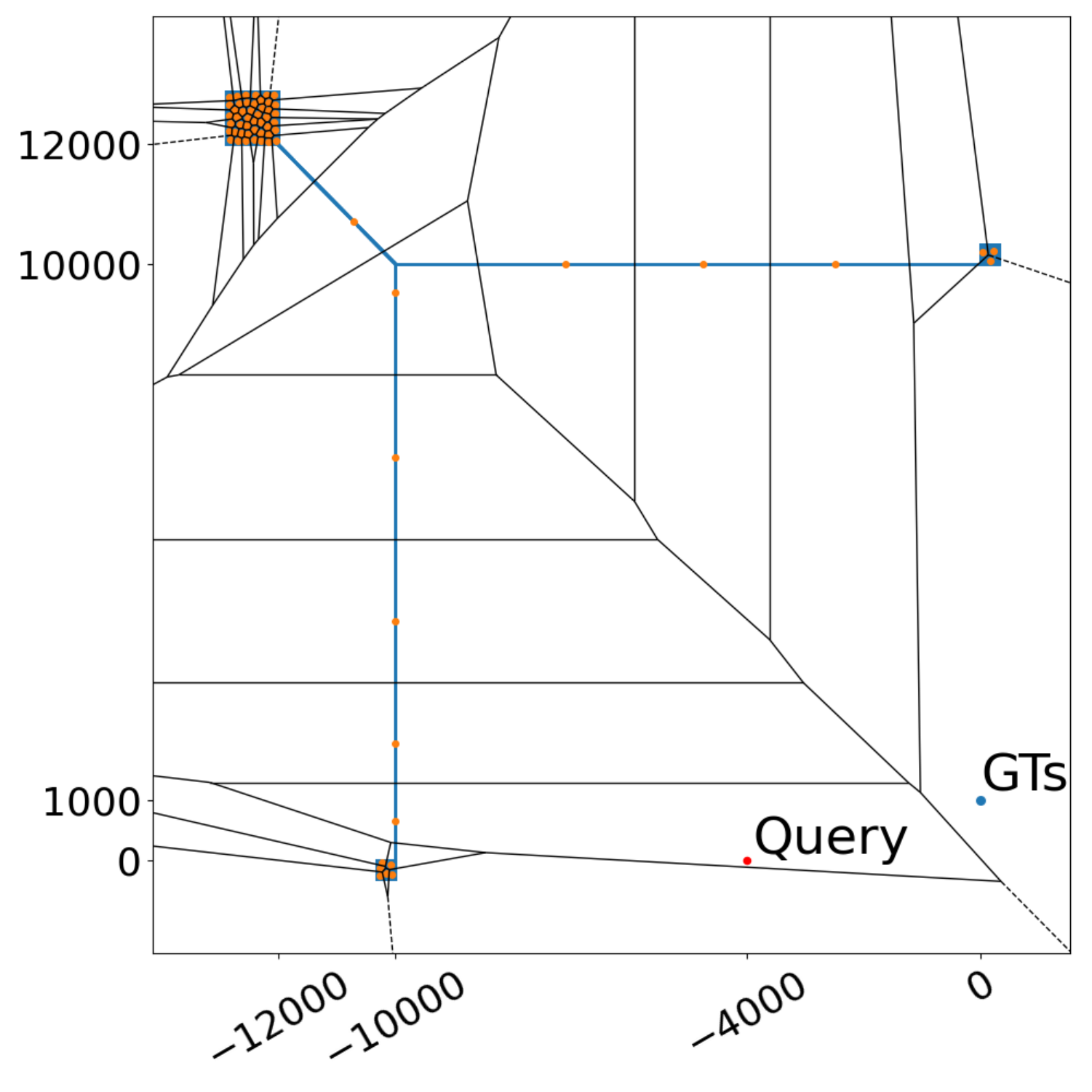}}%
    \subfigure[NSG: $K=256$ (Succeeded)]{%
        \includegraphics[clip, width=0.33\columnwidth]{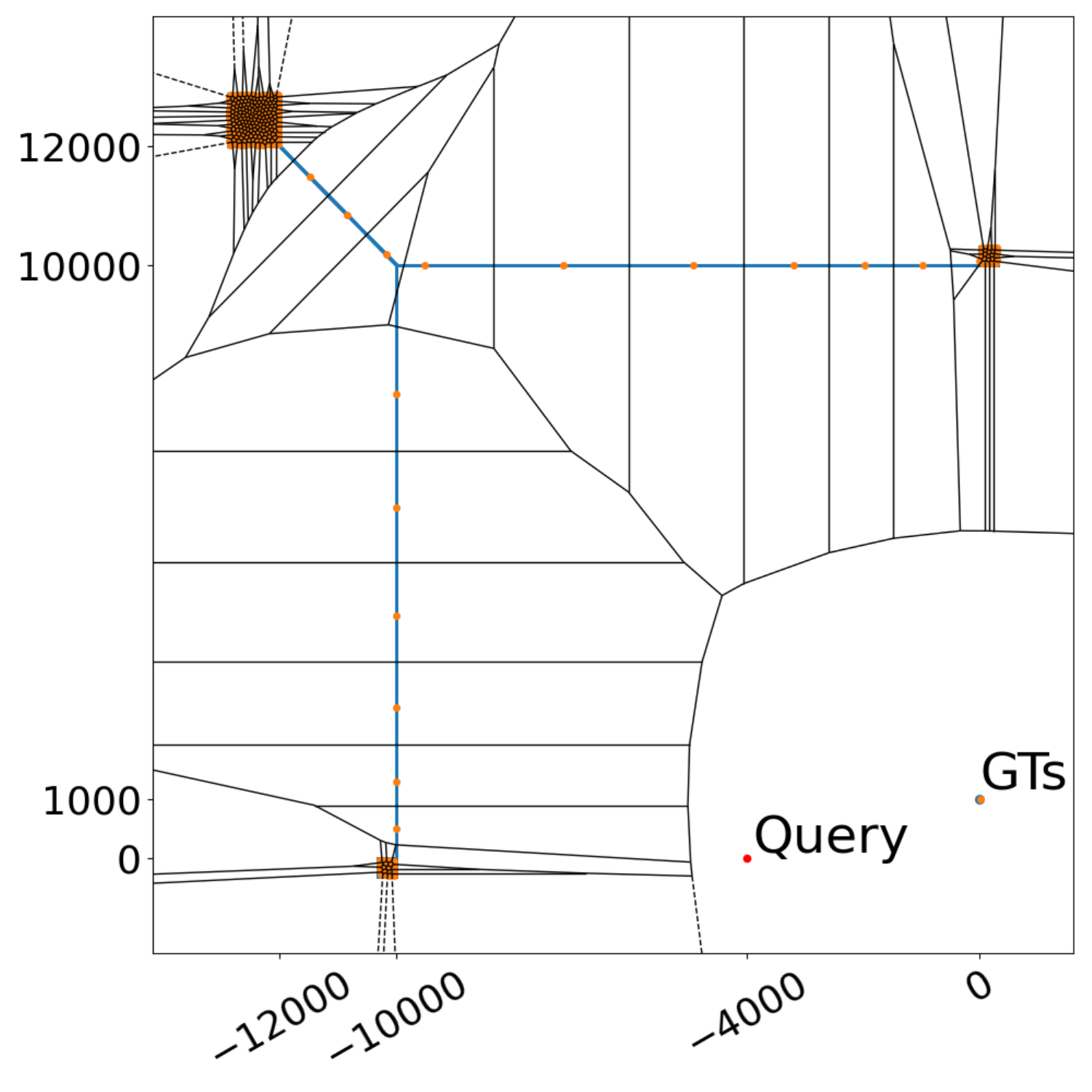}}%
    \subfigure[NSG: $K=512$ (Succeeded)]{%
        \includegraphics[clip, width=0.33\columnwidth]{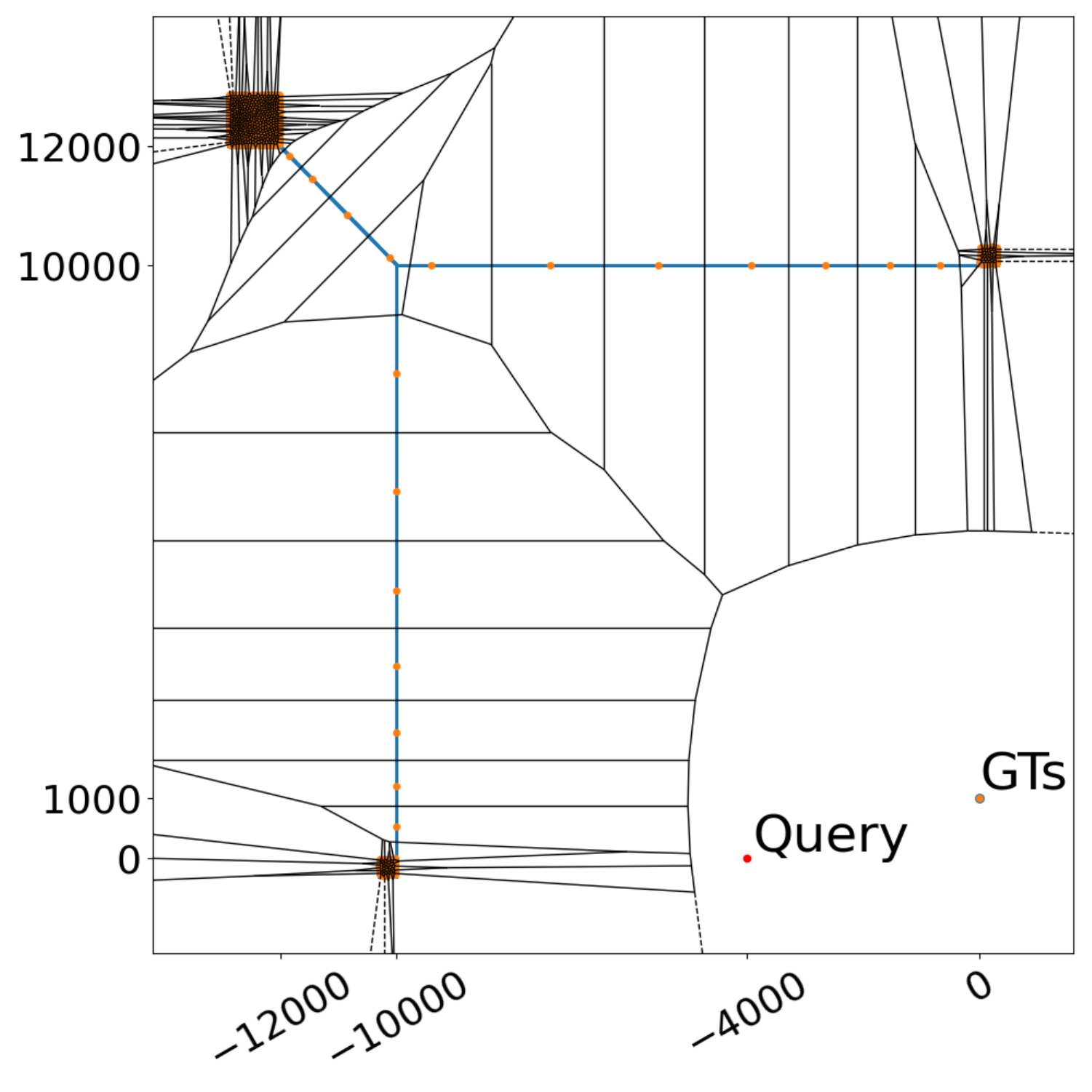}}%
    \\
    \subfigure[DiskANN: $K=32$ (Failed)]{%
        \includegraphics[clip, width=0.33\columnwidth]{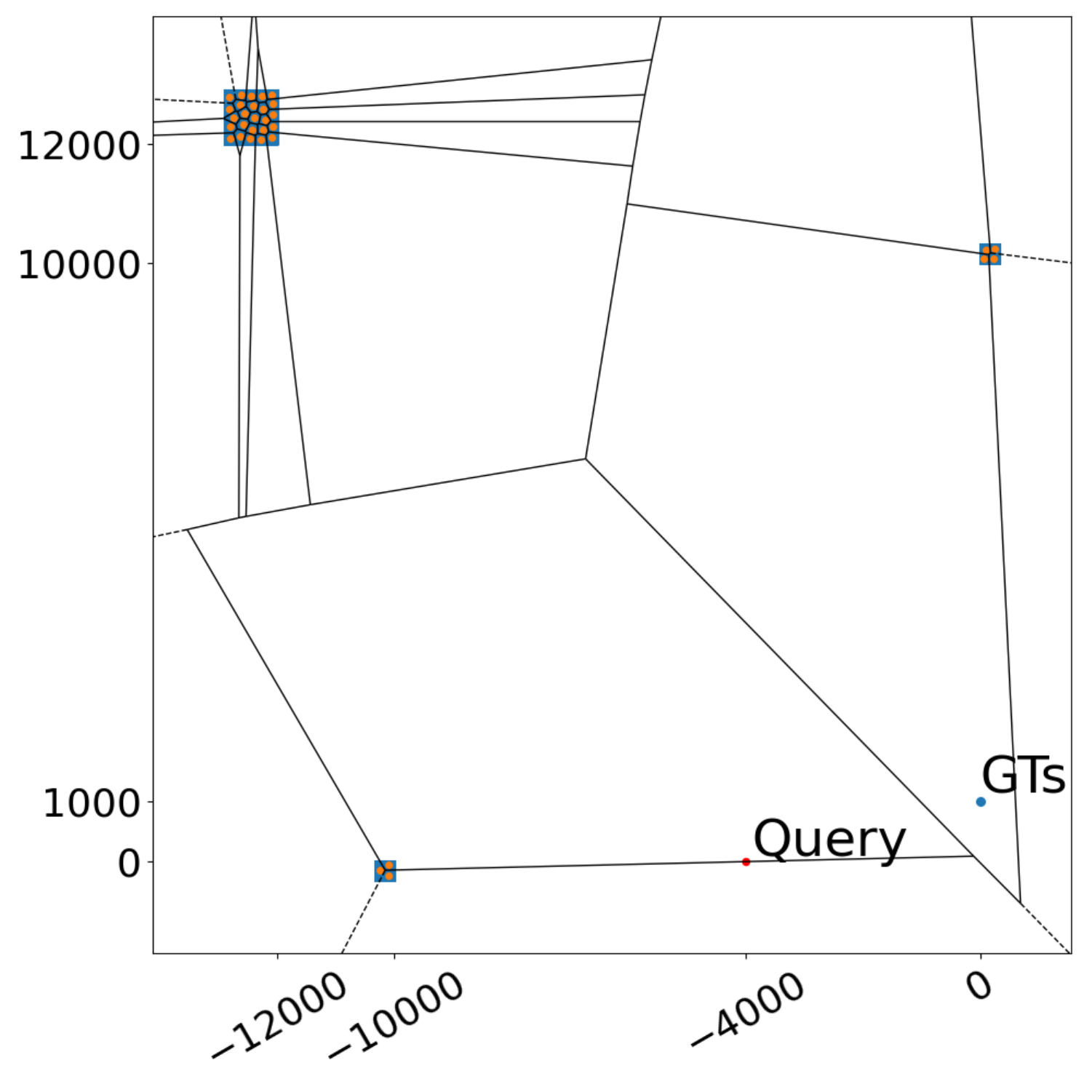}}%
    \subfigure[DiskANN: $K=128$ (Failed)]{%
        \includegraphics[clip, width=0.33\columnwidth]{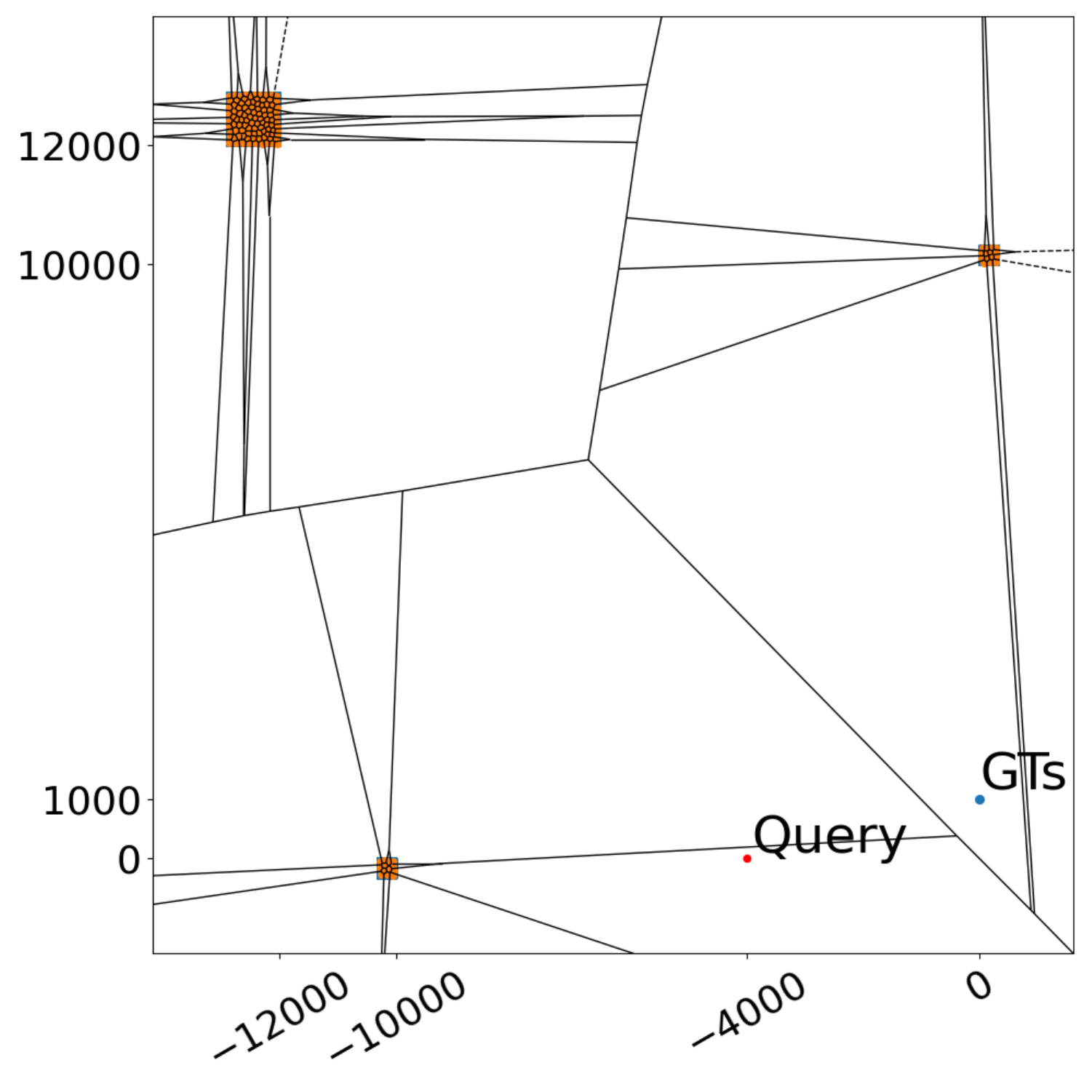}}%
    \subfigure[DiskANN: $K=256$ (Succeeded)]{%
        \includegraphics[clip, width=0.33\columnwidth]{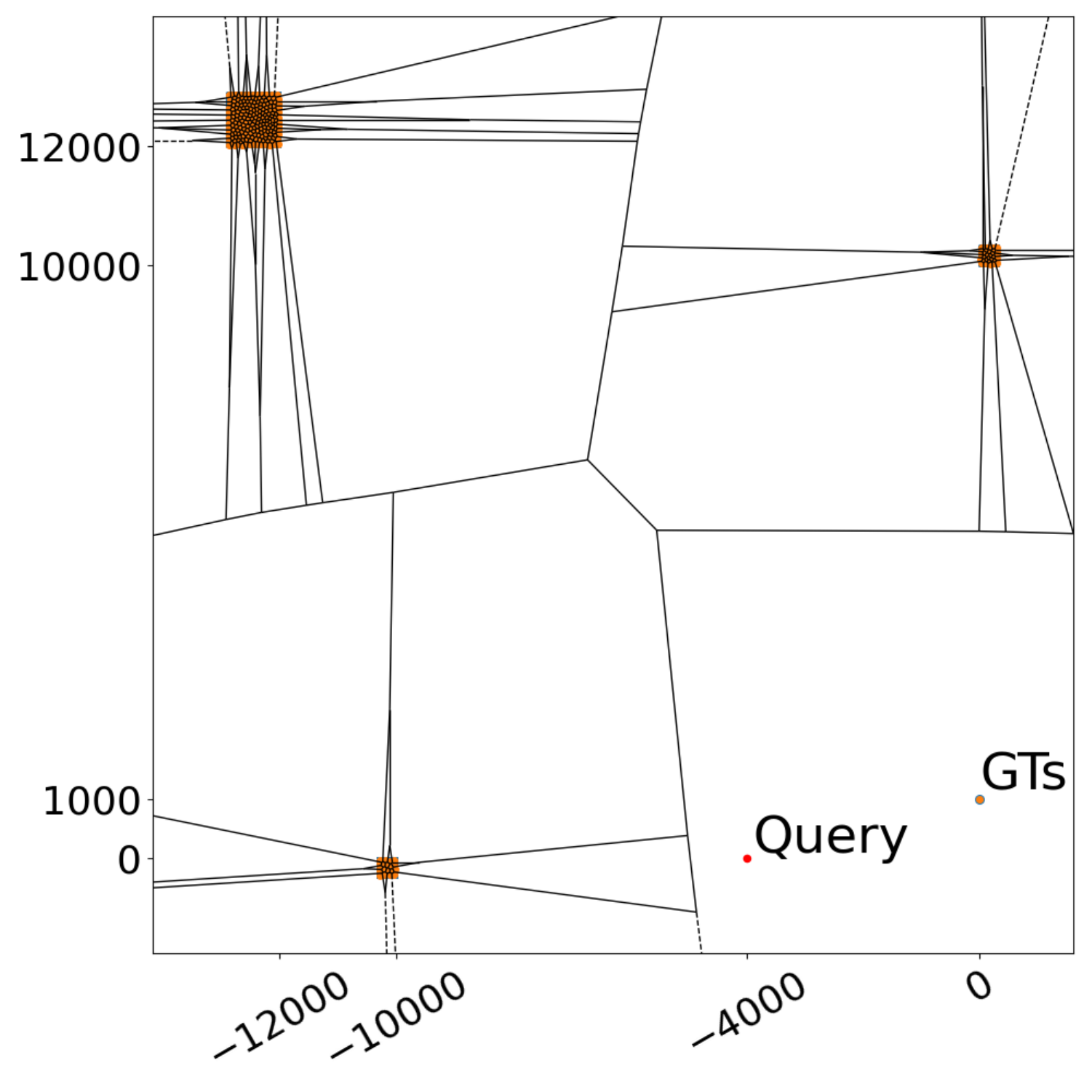}}%
    \caption{The Voronoi partition on the hard instance for NSG and DiskANN. (a) - (c) shows the results on NSG, and (d) - (f) shows the ones on DiskANN. For NSG, $K < 128$ leads to zero accuracy in search (failure), but $K\geq128$ leads to Recall@$10$ = 1.0 (success). For DiskANN, $K < 256$ leads to zero accuracy in search (failure), but $K\geq256$ leads to Recall@$10$ = 1.0 (success). Blue points show samples in the database $\mathcal{X}$. Orange points represent entry point candidates $\mathcal{D}$.}
    \label{fig:appendix_worst_voronoi_1}
\vskip -0.2in
\end{figure}



\end{document}